\documentclass[11pt,a4paper]{article}
\usepackage{tikz,todonotes}
\usepackage[english]{babel}
\usepackage[utf8]{inputenc}
\RequirePackage{filecontents}
\usepackage[T1]{fontenc}
\usepackage{epsfig}
\usepackage{graphicx}
\usepackage{bm}
\usepackage{amsfonts,amsmath,amsthm,amssymb}
\newtheorem{theorem}{Theorem}
\newtheorem{proposition}[]{Proposition}
\newtheorem{lemma}[]{Lemma}
\newtheorem{definition}{Definition}
\newtheorem{remark}{Remark}
\usepackage{tikz}
\usepackage{color,hyperref} 
\usepackage[margin=1.2in]{geometry}
\numberwithin{equation}{section}
\usepackage[toc,page]{appendix}

\usepackage{bbold}

\usepackage[normalem]{ulem}
\usepackage[backend=biber]{biblatex}
\usetikzlibrary{fit,
  matrix,
  positioning,
  decorations,
  decorations.pathreplacing
}

\usepackage{latexsym}

\newcounter{rhp}
\newenvironment{rhp}[1][]{\refstepcounter{rhp}\par\medskip
   \noindent \textbf{Riemann-Hilbert problem~\therhp. #1} \rmfamily}{\medskip}

\usepackage{xparse}
\usepackage{float,soul,ulem,cancel}
\providecommand*\email[1]{\href{mailto:#1}{#1}}
\DeclareMathOperator{\Tr}{Tr}

\def\&{\vspace{-5pt}&}

\usepackage[]{youngtab}

\usepackage{xcolor}
\definecolor{MyBlue}{rgb}{0.25,0.5,0.75}
\colorlet{NextBlue}{MyBlue!20}
\colorlet{SecondBlue}{MyBlue!40}

\newcommand{\bb}[1]{\mathbb{#1}}
\newcommand{\rh}[0]{RHP}
\newcommand{\pii}[0]{Painlev\'e II}
\newcommand{\tfn}[0]{$\tau$-function}
\newcommand{\cA}[0]{\mathcal{A}}
\newcommand{\cB}[0]{\mathcal{B}}
\newcommand{\cC}[0]{\mathcal{C}}
\newcommand{\cD}[0]{\mathcal{D}}

\newcommand{\wY}[0]{\widetilde{Y}}
\newcommand{\wF}[0]{\widetilde{F}}
\newcommand{\z}[0]{\zeta}
\newcommand{\x}[0]{\xi}
\newcommand{\s}[0]{\sigma_{3}}
\newcommand{\p}[0]{\psi}
\newcommand{\pri}[0]{\psi_{r}^{-1}}
\newcommand{\pr}[0]{\psi_{r}}
\newcommand{\prp}[0]{\psi_{r+}}
\newcommand{\prm}[0]{\psi_{r-}}
\newcommand{\pli}[0]{\psi_{l}^{-1}}
\newcommand{\pl}[0]{\psi_{l}}
\newcommand{\plp}[0]{\psi_{l+}}
\newcommand{\plm}[0]{\psi_{l-}}
\newcommand{\wP}[0]{\Psi}
\NewDocumentCommand{\tens}{t_}
 {%
  \IfBooleanTF{#1}
   {\tensop}
   {\otimes}%
 }
\NewDocumentCommand{\tensop}{m}
 {%
  \mathbin{\mathop{\otimes}\displaylimits_{#1}}%
 }

\begin{document}
\title{Fredholm determinant representation of the homogeneous \pii{} \tfn{}}

\author{ Harini Desiraju\footnote{{\email{harini.desiraju@sissa.it}}}
  \\
  \normalsize\it Scuola Internazionale Superiore di Studi Avanzati,   \\
  \normalsize\it Via Bonomea, 265, 34136 Trieste, Italy.} 

\date{\vspace{-5ex}}
\maketitle
\begin{abstract}
We formulate the generic $\tau$-function of the homogeneous Painlev\'e II equation as a Fredholm determinant of an integrable (Its-Izergin-Korepin-Slavnov) operator. The \tfn{} depends on the isomonodromic time $t$ and two Stokes parameters. The vanishing locus of the \tfn{}, called the {\it Malgrange divisor} is then determined by the zeros of the Fredholm determinant.
 \end{abstract}

\tableofcontents
\pagebreak

\section{Introduction}
Painlev\'e equations describe isomonodromic deformations of certain meromorphic linear ordinary differential equations on $\mathbb{CP}^{1}$.
In the theory of isomonodromic deformations, the Jimbo--Miwa--Ueno $\tau$-function is defined in terms of a closed 1--form, $\omega_{JMU}$ \cite{JMU} by the formula
\begin{equation*} \label{jmu}
\delta \log\tau_{_{JMU}} := \omega_{_{JMU}},
\end{equation*}
where $\delta$ denotes the total differential with respect to the 'deformation' parameters.
In this paper, we study the isomonodromic \tfn{} of the second order   scalar nonlinear ordinary differential equation (ODE) in the complex domain of the form
\begin{equation} \label{pii}
\dfrac{d^2 }{d x^2}u(x) = x u(x) - 2u(x)^3 \, , \quad x \in \bb{R},
\end{equation}
called  the {\it homogeneous}  Painlev\'e II equation. It arises as a consistency (zero-curvature) condition for the following set of linear ODEs for the $2\times 2$ complex valued matrix $\Phi(\lambda, x)$ 
\begin{gather} \label{FNin}
\begin{array} {c} \dfrac{d\Phi}{d\lambda} =  \left[  -i\left( 4\lambda^2 + x + 2u^2  \right) \s + 4\lambda u \sigma_{1} - 2 v \sigma_{2} \right] \Phi, \\ \\
\dfrac{d\Phi}{dx} = \left[ -i\lambda \s + u \sigma_{1}  \right] \Phi, \end{array}
\end{gather}
where $v(x) = u_{x}$, and the Pauli matrices 
\begin{gather}
\sigma_{1} = \left( \begin{array}{cc}
0 & 1 \\ 1 & 0
\end{array} \right) \, , \, \sigma_{2} = \left( \begin{array}{cc}
0 & - i  \\ i & 0
\end{array} \right) \, , \, \sigma_{3} = \left( \begin{array}{cc}
1 & 0 \\ 0 & -1
\end{array} \right).
\end{gather}
The set of ODEs \eqref{FNin} are called the Flaschka-Newell (FN) Lax pair \cite{FN}. The ODE \eqref{FNin} has an irregular singularity at $\lambda = \infty$ of Poincar\'e rank 3 and thus exhibits Stokes phenomenon on six rays in the complex plane. 
The generalized monodromy data described by the Stokes matrices is encoded in the jump matrix $G(\lambda,x)\in GL(2,\bb{C})$, $x\in \mathbb{R}$, $\lambda \in \Sigma$ on the contour $\Sigma\in \mathbb{C}$. The inverse problem consists of reconstructing the function 
\begin{gather}
\Psi(\lambda, x) = \Phi(\lambda, x) e^{- \frac{4}{3} \lambda^3 - x \lambda}, \quad \lambda = \bb{C} \backslash \Sigma, \quad x\in \mathbb{R} 
\end{gather}
from the generalized monodromy data. Such a solution $\Psi$ is piece-wise defined on the Stokes sectors with jumps on the Stokes rays specified by Stokes matrices. This is achieved by solving the following Riemann--Hilbert problem (RHP)
\begin{equation}
\begin{split}
\label{RHP_I}
\Psi_{+}(\lambda,x)& = \Psi_{-}(\lambda,x) G(\lambda,x),\quad \lambda\in\Sigma\\
\Psi(\lambda,x)  &= \bb{1}+{\mathcal O}(\frac{1}{\lambda}), \, \, \textrm{as} \, \, \lambda \rightarrow \infty,
\end{split}
\end{equation}
where $\Psi_\pm$ indicate the boundary value of $\Psi$ from the left side and the right side respectively of the oriented contour $\Sigma$, and $G (\lambda,x)$ is piece-wise defined on each of the Stokes rays with appropriate constraints on the Stokes matrices  (see section \ref{section:2} for a detailed description). With this data, the Malgrange form  is defined as follows.
\begin{definition}
 The {\it Malgrange form} associated with the RHP \eqref{RHP_I} is  defined as \cite{Mal}
\begin{gather}
\omega_\Sigma= \int_{\Sigma} \frac{d\lambda}{2\pi i } \Tr\left[ \Psi_{-}^{-1} \frac{\partial \Psi_{-} }{\partial \lambda} \delta G G^{-1}  \right], \label{DEF}
\end{gather}
where $\delta$ denotes the total differential with respect to the isomonodromic parameter  $x$.
\end{definition}
Since  $\omega_\Sigma$,  is a closed one form in the space of isomonodromic parameters, one can define (locally) the corresponding $\tau_\Sigma$ function as 
\[
\delta \log \tau_{_{\Sigma}} = \omega_\Sigma\,.
\]
The  Malgrange  form  is  a logarithmic form  in  the  sense  that  it  has  only  simple  poles with integer residues.
The locus in the parameter space where the  RHP problem \eqref{RHP_I} becomes unsolvable is called the {\it Malgrange divisor} because (in the language of algebraic geometry) it  can be described locally as the zero level set of a locally analytic function. 

The general gist is that this local expression can be represented (in abstract terms) as a Fredholm determinant (see for example \cite{palmer}). A concrete realization of this local function (a \tfn{})  as a Fredholm determinant (possibly globally defined on an open dense set of parameters) is of practical interest since it potentially allows for numerical investigation of the Malgrange divisor.

This paper treads this line of approach by providing a concrete representation for the \tfn{} of Painlev\'e\ II in terms of a Fredholm determinant expressed via an explicit (albeit complicated) kernel. According with this general framework, the zeros of the \tfn{} 
indicate the points where the \rh{} (i.e, the inverse monodromy problem) is not solvable.

It is well known that certain special solutions of Painlev\'e equations have a Fredholm determinant representation \cite{ BD, BO, IIKS, TW, tracy1994level1,tracy1994level}. The recent works of Lisovyy, Cafasso, Gavrylenko \cite{CGL,GL} provide a method to formulate the isomonodromic \tfn{}s of general solutions of PIII, PV, PVI as Fredholm determinants. There are two key aspects to their construction. One is the property that the \rh{}s of these Painlev\'e equations can be reduced on to a \rh{} on the unit circle. The second feature is that the jump on the unit circle 
enables the formulation of the \tfn{} as 
"Widom constant". An important feature of their construction is that the local parametrices of the \rh{} of the Painlev\'e equations are  described by known special functions which in turn act as `building blocks' of the \tfn{}.  For example, the local parametrices of the Painlev\'e VI \rh{} are given by hypergeometric functions and the \tfn{} is expressed as a Fredholm determinant of a hypergeometric kernel.

A natural question then is whether the \tfn{}s of Painlev\'e I, II, IV admit a Fredholm determinant representation. In a first step to answer this question in the case of Painlev\'e II, the present author recently 
showed that the \rh{} corresponding to the special 1-parameter (Ablowitz-Segur) family of solutions to the \pii{} equation \cite{AS} can be recast as a \rh{} on the imaginary axis as opposed to the unit circle in \cite{CGL}, hinting at a similar structure for the general \rh{} of \pii{} \cite{HD}. As a consequence, the corresponding \tfn{} (which is known to be the determinant of the Airy kernel \cite{TW}), can be formulated as a Widom constant. 

In the case of the \rh{} of \pii{}, it is known that under particular transformations that facilitate asymptotic analysis pf the Painlev\'e II transcendent at $x\rightarrow - \infty$, the local parametrices are described by parabolic cylinder functions $D_{\nu}(z)$ which we recall in Section \ref{section:2}. We then reduce the \rh{} to a \rh{} with a discontinuity on the imaginary axis in Section \ref{section:3}. However, we will see that the jump on the imaginary axis does not admit a Birkhoff factorization and hence the technique to construct Fredholm determinants in \cite{HD} is not applicable to our case. Instead, we use a variation of the formalism in \cite{MalB} namely, a lower, diagonal, upper triangular (LDU) factorization of the jump matrix 
to construct the \tfn{} as a Fredholm determinant of an integrable (IIKS) \cite{Deift,IIKS} operator with the parabolic cylinder functions acting as the 'building blocks'. 
In order to formulate our result, let
\begin{equation}
t= (-x)^{3/2}>0. \label{tdef123}
\end{equation}
We encode the Stokes parameters  $s_1$ and $s_3$ that define the \pii{}  RHP,   (see \eqref{pr_par}, \eqref{Stokes} below)  by :
\begin{gather}
\label{hnu}
\nu = -\frac{1}{2\pi i } \log(1-s_{1}s_{3}), \quad
h = -\frac{\sqrt{2\pi}}{ \Gamma(-\nu) s_3} e^{i\pi \nu},
\end{gather}
with $1-s_1s_3\neq 0$, $\arg \left(1-s_{1} s_{3}\right) \in \left( -\pi, \pi \right)$.
\begin{theorem}\label{theorem:1}
The \tfn{} of \pii{} equation can be expressed in terms of a Fredholm determinant of an integrable operator $\widetilde{\mathcal{K}}$ as follows
\begin{gather}
d_{t} \log \tau_{_{PII}} =  d_{t} \log \det\left[ \bb{1}_{_{L^2 (i\bb{R})}}- \widetilde{\mathcal{K}} \right] -\left[ \frac{4i\nu}{3} + \frac{2\nu^2}{t} \right] + \mathcal{F}(t, \nu, h),\label{tau_pii}
\end{gather}
where  $\mathcal{F}(t,\nu,h)$ is a regular function of the parameters $t$, $h$ and $\nu$ defined in \eqref{tdef123}, \eqref{hnu}.
The kernel of $ \widetilde{\mathcal{K}} $ takes the form
\begin{gather}
\widetilde{K}(z,w) = \frac{\cC(z,t)}{\cA(z,t)} \varphi_{+}^{2}(w) \int_{i\bb{R}+ \epsilon} \frac{d\widetilde{w}}{2\pi i} \frac{\varphi_{+}^{-2}(\widetilde{w})}{(z-\widetilde{w})(\widetilde{w}-w)} \cA(\widetilde{w},t) \cB(\widetilde{w},t) \label{ker1}
\end{gather} 
with the functions (see \eqref{def:JUMP_ABCD})
\begin{gather}
\cA(z,t) =  \z^{\nu} \x^{\nu} e^{\frac{2i}{3} t} \left(  e^{-\pi i \nu}  D_{-\nu}(i\z) D_{-\nu}(i\x) + \nu^2 h^{-4} e^{2\pi i \nu} D_{\nu-1}(\z) D_{\nu-1}(\x)  \right)  \\  \nonumber \\
\cB(z,t) =  \left(\frac{z+1/2}{z-1/2} \right)^{2\nu} \z^{\nu} \x^{-\nu} \left( i h^2 e^{- i \pi \nu}  D_{-\nu}(i\z) D_{-\nu-1}(i\x) + \nu h^{-2} e^{2\pi i \nu} D_{\nu-1}(\z) D_{\nu}(\x)  \right) \nonumber \\
\end{gather}
where $D_{\nu}$ is the parabolic cylinder function,  $\zeta=\zeta(z,t)$, $\xi=\xi(z,t)$ are given by
\[
\zeta \equiv \zeta(z,t) = 2t^{1/2} \sqrt{- \frac{4i}{3} z^3 + iz - \frac{i}{3}}; \quad \xi  = \zeta(-z,t)
\]
with branch cuts $(-\infty, -1]$ and $[1, \infty)$ respectively, and the function $\left(\frac{z-z_{-}}{z-z_{+}}  \right)^{\nu}$ is defined on $\mathbb{C}\backslash\left[z_{-}, z_{+}  \right]$ and the branch is fixed by the following asymptotic condition for $z\rightarrow \infty$ 
\begin{gather}
    \left(\frac{z-z_{-}}{z-z_{+}}  \right)^{\nu} \rightarrow \bb{1}.
\end{gather}
Moreover,
\begin{gather}
\cC(z,t) = \cB(-z,t) \, ; \quad
\varphi_{+}(w,t) = \int_{i\bb{R}- \epsilon} \frac{dw'}{2\pi i} \frac{\log\cA(w',t)}{w-w'}.
\end{gather}
From \eqref{def:JUMP_ABCD}, one can see that the functions $\mathcal{A}, \mathcal{B}, \mathcal{C}$ are analytic on a strip on the imaginary axis and 
\begin{gather}
    \lim_{z\rightarrow \pm i\infty} \mathcal{A}(z,t) = 1, \quad  \lim_{z\rightarrow \pm i\infty} \mathcal{B}(z,t) = 0, \quad \lim_{z\rightarrow \pm i\infty} \mathcal{C}(z,t) = 0.
\end{gather}
\end{theorem}

Some  comments 
are in order.
\begin{itemize}
\item[1.] The \tfn{} \eqref{tau_pii} is defined on $\bb{C}\backslash\left\lbrace 0\right\rbrace$ and is analytic in $t$. Refer to Remark:\ref{Remark_t} for the details. 

\item[2.] The important point of theorem \ref{theorem:1} is that the zeros of $\tau_{_{PII}}$, called the {\it Malgrange divisor} are determined solely by the zero locus of the Fredholm determinant. The Malgrange divisor is in one to one correspondence with the poles of the Painlev\'e II transcendent.

\item[3.]  The Malgrange divisor could be calculated  numerically by computing the Fredholm determinant in \eqref{tau_pii} employing the algorithm developed in \cite{bornemann}. 

\item[4.] The  limit from the general \tfn{} of \pii{} in \eqref{tau_pii} to the \tfn{} of the Hastings-McLeod family of solutions (determinant of the Airy kernel) \cite{TW} is singular because $s_3s_1=1$ ($\nu$ in \eqref{hnu} goes to infinity).

\end{itemize}

Theorem \ref{theorem:1}  is the first step in deriving the Fourier series representation of the Painlev\'e II tau-function obtained  in \cite{ILP}. A series representation of the \tfn{} \eqref{tau_pii} can be obtained from the minor expansion of the Fredholm determinant on an appropriate basis. A similar computation for the case of the Ablowitz-Segur solutions of \pii{} is worked out in \cite{HD}. Furthermore, we expect that the methods developed in this manuscript can be applied to some solutions of the Painlev\'e I and IV equations.
Finally, the study of the {\it inhomogeneous}  Painlev\'e II equation (\cite{Fokas}, Ch.5)
   \begin{equation}
    u_{xx} = x u - 2 u^3 + \alpha \label{inhom}
   \end{equation} 
  requires a significant extension of the techniques developed for  the {\it homogeneous}  Painlev\'e II equation  \eqref{pii}.
  Indeed the parameter $\alpha$ which we set to zero in \eqref{pii} induces a monodromy at the origin and the ideas developed  in this manuscript need a  non-trivial generalization. The \tfn{} of the general solution of  \eqref{inhom} has been expressed as a ratio of Hankel determinants in \cite{joshi2004generating}, \cite{DetPII}. Finally, we remark that the \tfn{}s of rational solutions of Painlev\'e equations have an interpretation not as Fredholm determinants, but as determinants of some special polynomials \cite{clarkson}. 
 \vspace*{0.3cm}

\noindent {\bf Acknowledgements.}\\ 
Thanks are due to Alexander Its, Marco Bertola, Oleg Lisovyy, Pasha Gavrylenko, and Tamara Grava for illuminating discussions and valuable suggestions. 
Part of the work was completed during visits to the Department of Mathematics and Statistics at Concordia University, and the Department of Mathematical sciences at IUPUI, which were supported by H2020-MSCA-RISE-2017 PROJECT No. 778010 IPADEGAN, INFN Iniziativa Specifica GAST.



\section{Setup}\label{section:2}

We recall the \rh{} associated to the Flaschka-Newell Lax pair \eqref{FNin} from \cite{Fokas} (see also \cite{Bot}). The matrix $\Psi(\lambda, x) \in GL(2, \bb{C})$, $x\in \mathbb{R}$  satisfies the following \rh{} on the contour in fig. \ref{fig:1}. 
\begin{rhp} \label{rhp:1}
\begin{itemize}
\item $\Psi(\lambda, x)$ is piece-wise analytic for $\lambda\in \bb{C}\backslash \cup_{k=1}^{6}\gamma_{k}$, \begin{gather}
\gamma_{k} = \left\lbrace \lambda \in \mathbb{C} : \arg \lambda =\frac{\pi}{6} + \frac{\pi}{3}(k-1)  \right\rbrace ,\,\, k=1,...,6.
\end{gather}
We define $\Psi_{k}:=\Psi(\lambda, x)\vert_{\Omega_{k}}$ with the Stokes sector $\Omega_{k}$ defined by
 \begin{equation}
\Omega_{k} = \left\lbrace \lambda \in \mathbb{C} : \frac{\pi}{6} (2k-3) < \arg \lambda < \frac{\pi}{6} (2k-1)   \right\rbrace , \, \, k=1,...,6.
 \end{equation} 
 
\item For $\lambda \in \gamma_{k}$, the following boundary condition is satisfied. \begin{equation}
\Psi_{k+1} = \Psi_{k} S_{k}, \quad \Psi_{7} =\Psi_{1}, \label{jump1}
\end{equation}
where the matrices $S_{k}$ are alternatively lower or upper triangular \begin{gather}
S_k = \left( \begin{array}{cc}
1 & 0 \\ s_k e^{2i\theta(\lambda, x) } & 1
\end{array}  \right) \, \,  \textrm{for} \,\, k \equiv 1\bmod 2 , \quad S_k = \left( \begin{array}{cc}
1 & s_k e^{-2i\theta(\lambda, x)} \\ 0 & 1
\end{array}  \right) \, \,  \textrm{for} \,\, k \equiv 0 \bmod 2, \label{2.2}
\end{gather}
and the exponent $\theta(\lambda, x) = \frac{4}{3}\lambda^3 + x \lambda$. The Stokes parameters $s_k$ are constants satisfying the constraint
\begin{equation}
s_{k+3} = -s_k \,\, , \, \, s_{1} - s_{2} + s_{3} + s_{1} s_{2} s_{3} =0, \label{par}
\end{equation}  
and the jump matrices satisfy the following identity
\begin{gather}
S_{1} S_{2}... S_{6} = \mathbb{1}. \label{JUMPIDEN}
\end{gather}
\item In the asymptotic limit of $\lambda$, \begin{gather}\lim_{\lambda \rightarrow \infty} \Psi(\lambda,x)  = \bb{1}. \label{asym1}
\end{gather} 
\end{itemize}
\end{rhp}
\begin{figure}[H] 
\includegraphics[trim={8cm 8.9cm 0.8cm 9cm}, clip, width=12cm]{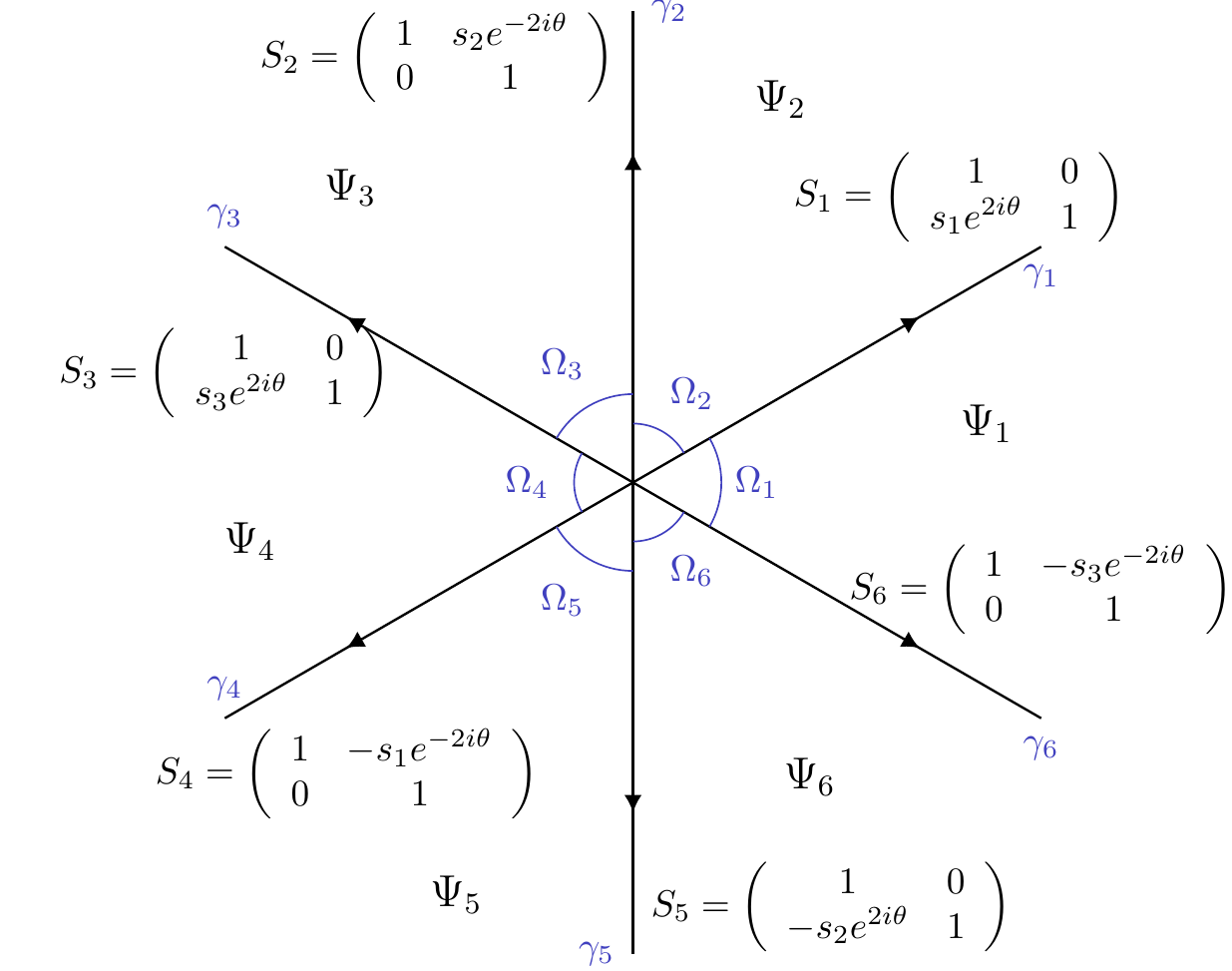}
\centering
\caption{Stokes rays \label{fig:1}}
\end{figure}
The constraint on Stokes data \eqref{par} implies that the solution $\Psi(z,t)$ depends only on two Stokes parameters. We will see in the next subsection that all the functions depend on $s_{1}$, $s_{3}$. In this paper, we are concerned with the generic 2-parameter solutions\footnotetext[2]{The Ablowitz-Segur family of solutions correspond to the case $s_{2}=0$.} that correspond to the following constraints on the Stokes parameters
\begin{gather}
s_{1} s_{3} \neq 1 \, ; \quad \arg(1-s_{1} s_{3}) \in (-\pi, \pi).
\end{gather}
In order to modify the Riemann--Hilbert contour of \pii{}, we perform the change of variables ($x\in \mathbb{R} $)
\begin{gather}
\lambda = (-x)^{1/2} z , \quad t= (-x)^{3/2}>0. \label{cov}
\end{gather}
The characteristic exponent $\exp\left(i \theta(\lambda,x) \right)$ in \eqref{2.2} is then replaced by
\begin{equation}
e^{it\theta(z)}, \quad \theta(z) = \frac{4}{3} z^3 - z. 
\end{equation}
The stationary points are then $z_{\pm} = \pm 1/2$. 
The contour in fig. \ref{fig:1} can be deformed into fig. \ref{fig:grl}
\begin{figure}[H] 
\includegraphics[trim={0.8cm 2cm 2.3cm 2cm},clip, width=10cm]{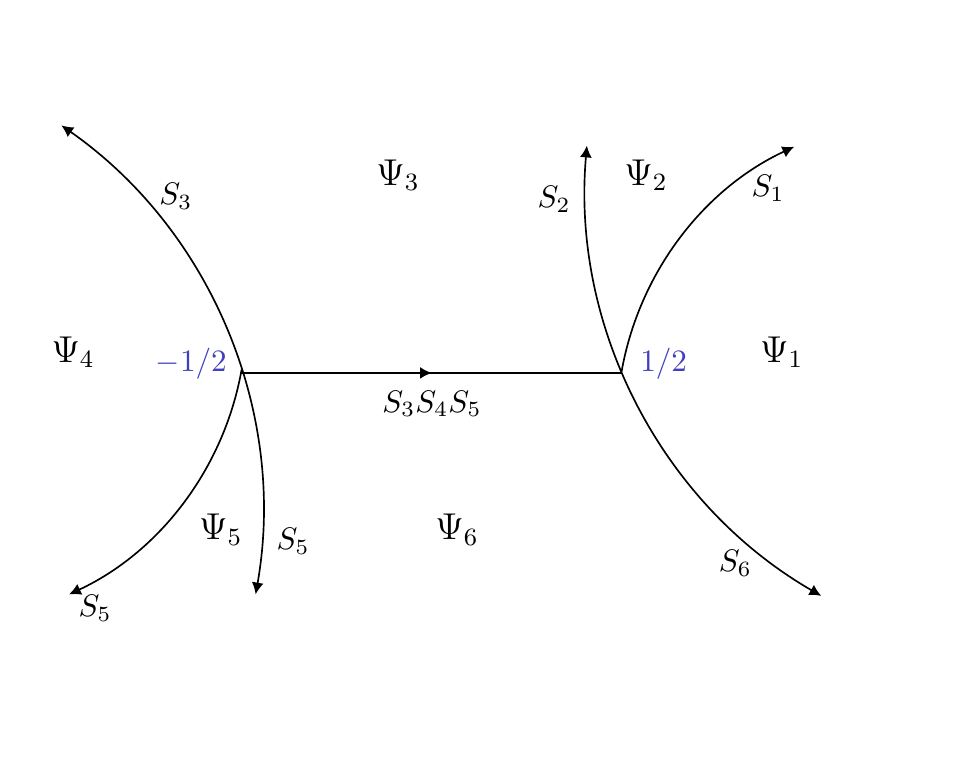}
\centering
\caption{Deforming the contour in fig. \ref{fig:1} \label{fig:grl}}
\end{figure}
Noticing that the product of Stokes matrices $(S_{3} S_{4} S_{5})^{-1}$ can be written as a product of lower triangular, diagonal and upper triangular matrices (LDU)
\begin{gather}
(S_{3} S_{4} S_{5})^{-1} = \left( \begin{array}{cc} 1-s_{1}s_{3} & s_{1} e^{-2 it \theta(z)} \\ s_1 e^{2it \theta(z)}& 1+ s_{1} s_{2} \end{array} \right) = S_{L} S_{D} S_{U} \nonumber \\
=  \left( \begin{array}{cc} 1 & 0 \\ s_1 (1- s_{1} s_{3})^{-1} e^{2it \theta(z)} & 1 \end{array} \right) \left( \begin{array}{cc} 1-s_{1}s_{3} & 0 \\ 0 & (1- s_{1} s_{3})^{-1} \end{array} \right) \left( \begin{array}{cc} 1 & s_{1}(1- s_{1} s_{3})^{-1} e^{-2it \theta(z)} \\ 0 & 1 \end{array} \right) ,
\end{gather}
the contour in fig. \ref{fig:1} can be transformed into  the  contour $\Sigma $  in fig. \ref{fig:2}.  One can easily check there is  no monodromy  around the points $z= \pm 1/2$. 
\begin{figure}[H] 
\hspace*{-1.4cm}
\includegraphics[trim={0.8cm 2cm 0.4cm 2cm},clip, width=12cm]{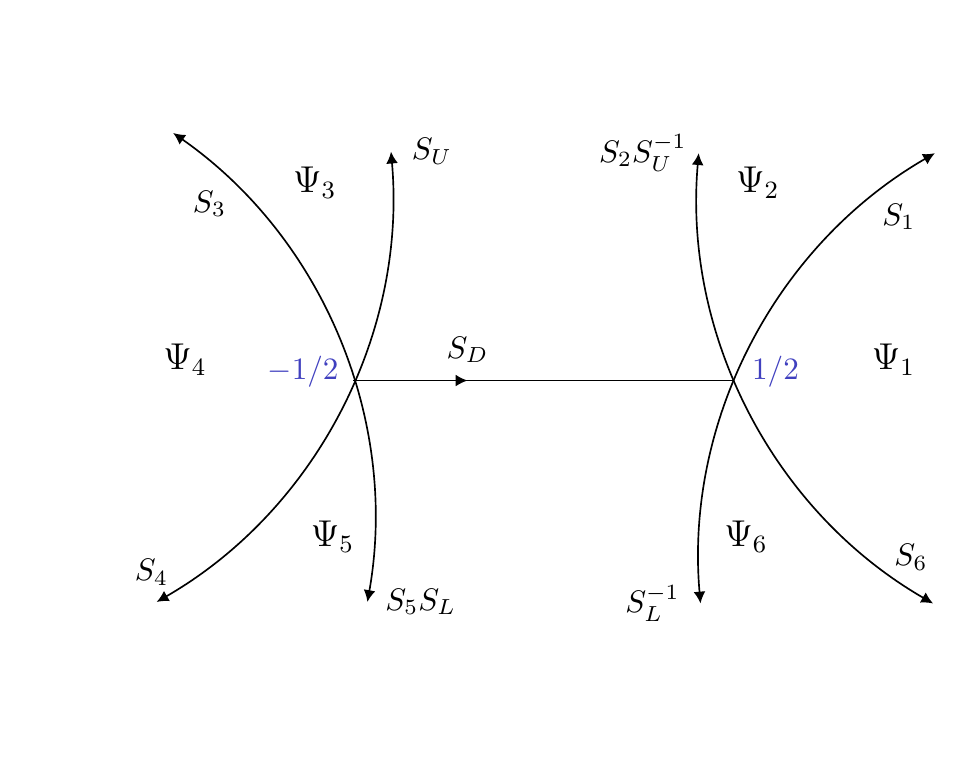}
\centering
\caption{Deformed \pii{} Riemann--Hilbert contour $\Sigma$.  \label{fig:2}}
\end{figure}
In fig. \ref{fig:2}, $S_{D} = (1-s_{1} s_{3})^{\s}$. Defining 
\begin{gather} \label{NU}
\nu = -\frac{1}{2\pi i } \log(1-s_{1} s_{3}),
\end{gather}
one can verify that the function 
\begin{gather}
\Psi^{D}(z,t) = \left( \frac{z-z_{-}}{z-z_{+}} \right)^{\nu \sigma_{3}} 
\end{gather}
satisfies the \rh{} on the segment $\left[ z_{-}, z_{+} \right]$ where $z_{\pm } = \pm 1/2$ and the branch is fixed by the following asymptotic condition for $z\rightarrow \infty$ 
\begin{gather}
    \left(\frac{z-z_{-}}{z-z_{+}}  \right)^{\nu} \rightarrow \bb{1}.
\end{gather}
On the contour $\Sigma$ in fig. \ref{fig:2}, the function $\wP(\lambda(z),x) \equiv \wP(z,t)$ in \eqref{FNin} solves the following \rh{}
\begin{rhp} \label{rhp:4}
\begin{itemize}
 \item $\Psi(z,t)$ is analytic on $z\in \bb{C}\backslash \Sigma $, $t\in \mathbb{R}$.
 \item 
 For $z\in \Sigma$, on each of the Stokes rays
\begin{gather}
 \wP_{-}^{-1}(z,t) \, \wP_{+}(z,t) =  G(z,t),
\label{J_Gamma}
\end{gather}
where $G(z,t)$ is piece-wise defined on each of the rays of the contour $\Sigma$, $\wP_{\pm}$ are  the boundary values of $\Psi$ from the left side and the right side of the oriented contour $\Sigma$ respectively, and the constraint  \eqref{JUMPIDEN} holds.
\item $\lim_{z\rightarrow \infty} \Psi(z,t) = \mathbb{1} $
\end{itemize}
\end{rhp}

In terms of the RHP \eqref{J_Gamma}, the \tfn{} \eqref{DEF} is 

\begin{equation}\label{def:1}
\delta\log\tau_{_{PII}} \equiv \delta\log\tau_{_{\Sigma}}(t) := \int_{\Sigma} \frac{d z}{2\pi i} \Tr\left[ \wP_{-}^{-1}\wP_{-}' \dot{G} G^{-1}  \right], 
\end{equation}
where $\dot{G}$ means the {\it total} derivative with respect to $t$  and $\wP_{-}' $ means partial derivative with respect to $z$.
We will  convert the above expression to a   Fredholm determinant in theorem \ref{theorem:1}.

\subsection{Parametrices}\label{subsection:2.1}

To express the \tfn{} \eqref{def:1} in terms of a Fredholm determinant we need to construct a  ``parametrix''  solutions, namely  ``local solutions'' of the \rh{}. These local solutions are patched together  and  the actual problem  can be recast as the solution of a  compact (trace--class)  perturbation of the identity.

The effectiveness of the idea relies entirely upon the level of simplicity of these parametrices; the simpler (or rather,  more explicit) these reference parametrices are, the more practical the approach is in studying the final problem.  

Keeping this in mind, in this section we construct an explicit solution to a Riemann--Hilbert problem to be used as parametrix for the final one. To this end we recall from (\cite{Fokas}, Ch.9 pg.318) the construction of the local parametrices of Painlev\'e II \rh{} in fig. \ref{fig:2}(the left and right parametrices around the points $z= \pm 1/2$ respectively), in terms of parabolic cylinder functions \cite{Fokas}. 
\subsubsection{Model problem}
Let $Z(\zeta)$ be a $2\times 2$ matrix valued function that solves the following \rh{}.
\begin{rhp} \label{rhp:3}
\begin{itemize}
\item $Z (\zeta)$ is a piece-wise holomorphic function defined as follows in each sector shown in fig. \ref{fig:3}
\begin{gather} \label{Zi}
Z (\zeta) = 
\begin{cases}Z_{0} (\zeta), \quad \arg\zeta \in \left(-\frac{\pi}{4},0 \right) \\ 
Z_{1} (\zeta), \quad \arg\zeta \in \left(0, \frac{\pi}{2} \right) \\
Z_{2} (\zeta), \quad \arg\zeta \in \left(\frac{\pi}{2}, \pi \right)  \\
 Z_{3} (\zeta), \quad \arg\zeta \in \left(\pi , \frac{3\pi}{2} \right) \\
 Z_{4} (\zeta), \quad \arg\zeta \in \left(\frac{3\pi}{2},\frac{7\pi}{4} \right). \end{cases}
\end{gather}
Under the transformation $\zeta \rightarrow -\zeta$ the following symmetry relation holds
\begin{equation}
\sigma_{3} Z_{k+2}\left( e^{i\pi} \zeta \right) \sigma_{3} = Z_{k}(\zeta) e^{-i\pi(\nu+1) \sigma_{3}}. \label{Par_symm}
\end{equation} 

\begin{figure}[H] \label{fig:3}
\includegraphics[trim={0 2cm 0 2cm},clip, width=8.5cm]{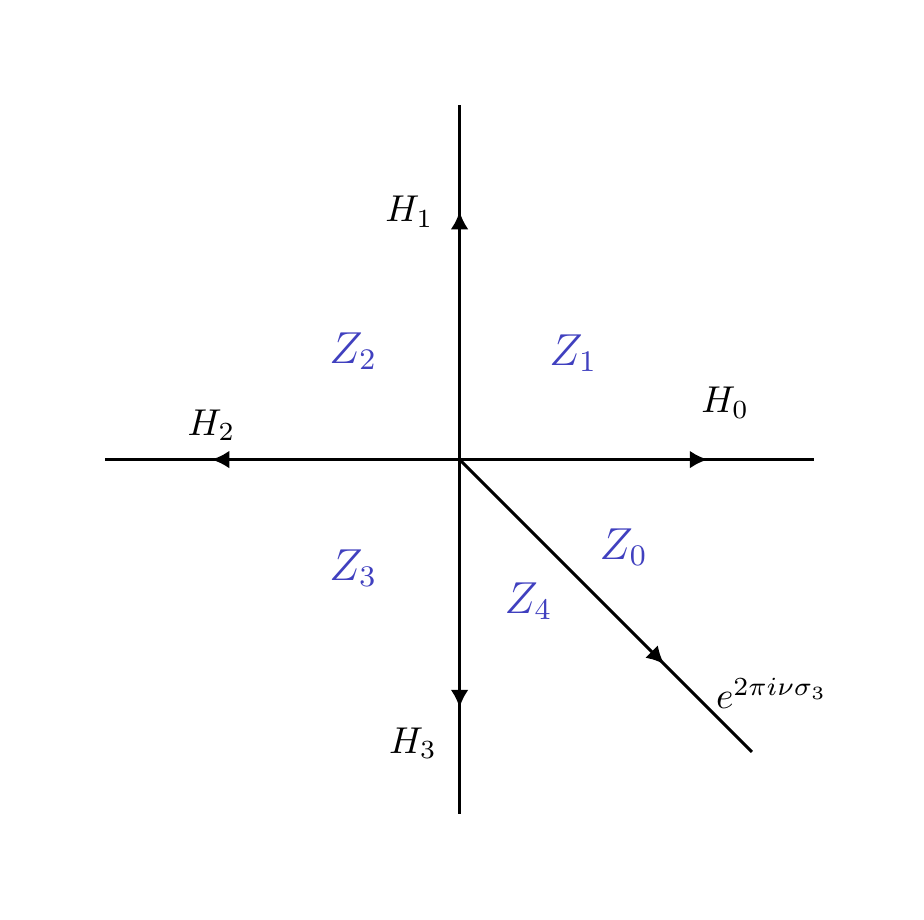}
\centering
\caption{Riemann--Hilbert contour of parabolic cylinder function.}
\end{figure}
\item In each sector, the  following jump conditions are satisfied
\begin{equation}
Z_{k+1}(\zeta) = Z_{k} (\zeta) H_{k}, \quad \arg\zeta = \frac{\pi}{2}k, \quad k=0,1,2,3,4,  \label{h1}
\end{equation}
and $Z_{5} =Z_{0}$.
The jump matrices 
\begin{gather}
\begin{array}{c}
 H_{0 } = \left( \begin{array}{cc} 1 & 0 \\ h_{0} & 1 \end{array} \right), \quad H_{1} =  \left( \begin{array}{cc} 1 & h_{1} \\ 0 & 1 \end{array} \right) \, ,  \quad H_{4}\equiv H_{D}= e^{2\pi i \nu \s}; \\  \\
 H_{k+2} = e^{i\pi \left( \nu + \frac{1}{2} \right) \sigma_{3}} H_{k} e^{-i\pi\left( \nu + \frac{1}{2} \right)\sigma_{3}}, \quad \textrm{for} \, k=0,1. 
 \end{array}\label{pc_jump}
\end{gather}
The Stokes parameters $h_{0}$ and $h_{1}$ are defined as follows
\begin{gather}
h_{0} = -i \frac{\sqrt{2\pi}}{\Gamma (\nu+1)}, \quad h_{1} = \frac{\sqrt{2\pi}}{\Gamma (-\nu)} e^{i\pi\nu}, \quad 1 + h_{0}h_{1} = e^{2\pi i \nu}, \label{h2}
\end{gather}
and the identity $e^{2\pi i \nu\sigma_{3}}H_{0}H_{1}H_{2}H_{3} = I$ implies the triviality of the monodromy at the origin.
\item As $\zeta \rightarrow \infty $,
\begin{equation}
Z(\zeta) = \zeta^{-\sigma/2} \frac{1}{\sqrt{2}} \left(\begin{array}{cc} 1 & 1 \\ 1 & -1 \end{array} \right) \left( \bb{1} + O(\zeta^{-2}) \right) e^{\left( \frac{\zeta^2}{4} - (\nu + \frac{1}{2})\log\zeta \right)\sigma_{3}}. \label{asymp_PC}
\end{equation}
\end{itemize}
\end{rhp}
In the zeroth sector, $Z(\zeta)$ is expressed in terms of the Wronskian of the parabolic cylinder functions as 
\begin{equation}
Z_{0}(\zeta) = 2^{-\sigma_{3}/2} \left(\begin{array}{cc} D_{-\nu-1}(i\zeta) & D_{\nu}(\zeta) \\ \frac{d}{d\zeta} D_{-\nu-1}(i\zeta) & \frac{d}{d\zeta} D_{\nu}(\zeta)  \end{array}  \right)\left(\begin{array}{cc} e^{i\frac{\pi}{2}(\nu+1)} & 0 \\ 0 & 1 \end{array} \right). \label{Z_0} 
\end{equation}
The parabolic cylinder functions $D_{\nu}(z)$, $D_{-\nu-1}(iz)$ are independent solutions to the differential equation
\begin{gather}
\frac{d^2 y(z)}{dz^2} + \left( \nu + \frac{1}{2} - \frac{1}{4} z^2 \right) y(z) =0,
\end{gather}
and $\lim_{z\rightarrow \infty} D_{a}(z) e^{z^2/4} z^{-a} = 1$ for $\vert \arg(z) \vert <\frac{\pi}{2}$. Refer to Ch. 19, \cite{abramowitz1988handbook} for details. 

\subsubsection{Local parametrices}
Under the map
\begin{equation}\label{variablechange}
\zeta(z) = 2 \sqrt{-\frac{4it}{3}z^3 + itz - \frac{it}{3}}= \frac{4}{\sqrt{3}}e^{-3i \pi /4} \sqrt{t}\left( z-\frac{1}{2} \right)\sqrt{z+1},
\end{equation}
where $\zeta(z)$ has a branch cut on $(-\infty, -1]$, we define the right parametrix around $z_{+} = 1/2$ as
{\small \begin{gather}
\pr(z,t) = 
\left( \zeta(z) \frac{z-z_{-}}{z-z_{+}} \right)^{\nu\sigma_{3}} \left(-\frac{h_{1}}{s_{3}}  \right)^{-\sigma_{3}/2} e^{\frac{it}{3}\sigma_{3}} 2^{-\sigma_{3}/2} \left( \begin{array}{cc} \zeta(z) & 1 \\ 1 & 0 \end{array} \right) Z(\zeta(z)) \left(-\frac{h_{1}}{s_{3}}  \right)^{\sigma_{3}/2}e^{it \theta(z) \s}, \label{right_parametrix}
\end{gather}}
and the left parametrix around $z=-1/2$ is determined through the symmetry relation
\begin{gather}
\pl (z,t) = \sigma_{2} \pr(-z,t) \sigma_{2}. \label{left_parametrix}
\end{gather}
In \eqref{right_parametrix}, the parameter $\nu$ and $h_{1}$ are determined by the Stokes parameters $s_{1}$, $s_{3}$. Recall from \eqref{NU} and \eqref{h2},
\begin{gather}
\nu = - \frac{1}{2\pi i} \log(1-s_{1} s_{3}) \, ; \quad
 h_{1} = \frac{\sqrt{2\pi}}{\Gamma(-\nu)} e^{i\pi \nu}. \label{pr_par}
\end{gather}
\begin{figure}[H] \label{fig:6}
\hspace*{-0.8cm}
\includegraphics[trim={0 1cm 0 1cm},clip, width=16.5cm]{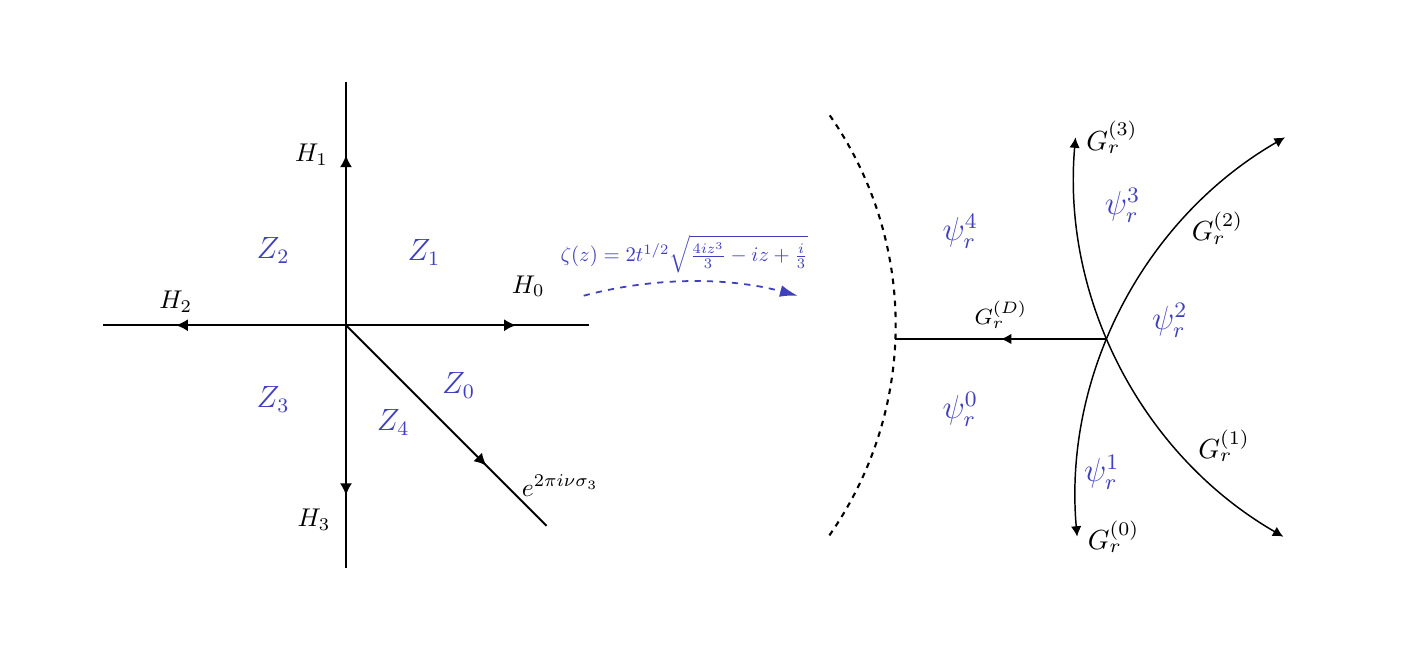}
\centering
\caption{Mapping the $\zeta$-plane to the right-half of $z$-plane}
\end{figure}
In each sector,
{\small \begin{gather}
\pr^{(k)}(z,t) =\left( \zeta(z) \frac{z-z_{-}}{z-z_{+}} \right)^{\nu\sigma_{3}} \left(-\frac{h_{1}}{s_{3}}  \right)^{-\sigma_{3}/2} e^{\frac{it}{3}\sigma_{3}} 2^{-\sigma_{3}/2} \left( \begin{array}{cc} \zeta(z) & 1 \\ 1 & 0 \end{array} \right) Z_{k}(\zeta(z)) \left(-\frac{h_{1}}{s_{3}}  \right)^{\sigma_{3}/2}e^{it \theta(z) \s}. \label{pr_sec}
\end{gather}}
The jumps on Stokes rays in the right and left half planes are denoted by
\begin{gather}
G_{r} := G(z,t) \vert_{\Re(z)>0}\, ; \quad G_{l}:= G(z,t) \vert_{\Re(z)<0}. \label{G_r}
\end{gather}
As a consequence of \eqref{left_parametrix},
\begin{gather}
G_{l}(z,t) = \sigma_{2} G_{r}(-z,t) \sigma_{2} \label{G_symm}
\end{gather}
We now establish the relation between the Stokes matrices of the parabolic cylinder functions $H_{i}$ in \eqref{pc_jump} and $G_{r}\equiv G_{r}^{(i)}$ in \eqref{G_r}. Introducing the notation
\begin{gather}
h = \left( -\frac{h_{1}}{s_{3}} \right)^{1/2},\label{Stokes}
\end{gather}
in a sector $k$ on the right half-plane in fig. \ref{fig:2}, $\pr$ satisfies the following jump condition
\begin{gather}
 \pr^{k+1}(z,t) = \pr^{k}(z,t) e^{-it \theta(z)\s}\left(-\frac{h_{1}}{s_{3}}  \right)^{-\sigma_{3}/2} Z_{k}^{-1} Z_{k+1} \left(-\frac{h_{1}}{s_{3}}  \right)^{\sigma_{3}/2} e^{it \theta(z)\s} \nonumber \\
= \pr^{k}(z,t) e^{-it \theta(z)\s}h^{-\s} H_{k} h ^{\s} e^{it \theta(z)\s} \nonumber \\
= \pr^{k}(z,t) G_{r}^{k}.   \label{H_i}
\end{gather}
Note that $Z_{5} = Z_{0}$ implies that $\pr^{5}(z,t) = \pr^{0}(z,t)$. Therefore, in terms of $H_{k}$, $G_{r}$ is 
\begin{gather}
G_{r}^{(k)}(z,t) =  e^{-it\theta \sigma_{3}} h^{-\sigma_{3}} H_{k} h^{\sigma_{3}} e^{it\theta \sigma_{3}}. \label{GH}
\end{gather}
We define the variable
\begin{equation}
\x (z,t) := \z(-z,t) \label{varminus}
\end{equation} 
with a branch cut on $[+1,\infty)$, that maps the $\x$-plane to the left half-plane of fig.\ref{fig:2} and a similar computation follows for the left parametrix due to the symmetry relation \eqref{left_parametrix}.  We denote the jump condition in each sector on the respective half-planes in fig. \ref{fig:2} by
\begin{gather}
\p_{r,l; +}(z,t) = \p_{r,l; -}(z,t) G_{r,l}(z,t). \label{G_rl}
\end{gather}
\begin{remark}
It is easy to check that the map in fig. \ref{fig:6} holds not just locally but asymptotically in the right half plane, and similarly, the transformation $\x(z,t)$ is valid in the left half plane. Therefore the functions $\psi_{r}(z,t)$ and $\psi_{l}(z,t)$ can indeed be defined on the right and left half planes respectively. 
\end{remark}

\begin{remark}\label{Remark_t}
The transformation \eqref{variablechange} is not valid at the point $t=0$. This implies that $\tau_{_{PII}}$ in \eqref{tau_pii} is valid for $t\in \bb{C}\backslash \lbrace 0 \rbrace$.  
\end{remark}

\section{Reduction to a \rh{} along the imaginary axis} \label{section:3}

Define a matrix function $\Theta(z,t)$ as a ratio of the global solution $\Psi$ on $\Sigma$ in \eqref{J_Gamma} and the local parametrices $\pr$ in \eqref{right_parametrix}, $\pl$ in \eqref{left_parametrix}.
\begin{gather}
\Theta(z,t) := \begin{cases} \wP(z,t) \pri (z,t); \,\, \Re(z) >0 \\  \wP(z,t) \pli (z,t); \,\, \Re(z)<0. \end{cases} \label{cR}
\end{gather}
Note that the local parametrices cancel the jump of the global parametrix on $\Sigma$, ensuring that the function $\Theta(z,t)$ has a jump only on the imaginary axis, solving the following \rh{}.
\begin{rhp} \label{rhp:5}
\begin{itemize}
\item $\Theta(z,t)$ is analytic on $z\in \bb{C}\backslash i\bb{R}$
\item For $z\in i\bb{R}$,
\begin{equation}
 \Theta_{+}(z,t) = \Theta_{-}(z,t) J(z,t) \label{J2}
\end{equation}
 where $J(z,t) =  \pr^{(0)}(z,t) \left[ \pl^{(4)}(z,t)\right]^{-1}$.

\item As $z\rightarrow \infty$,  $\Theta(z,t) = \bb{1} + \mathcal{O} \left( z^{-1} \right)$.
\end{itemize}
\end{rhp}
\begin{figure}[H] \label{fig:4}
\includegraphics[trim={0 2cm 0 2.2cm},clip, width=12.5cm]{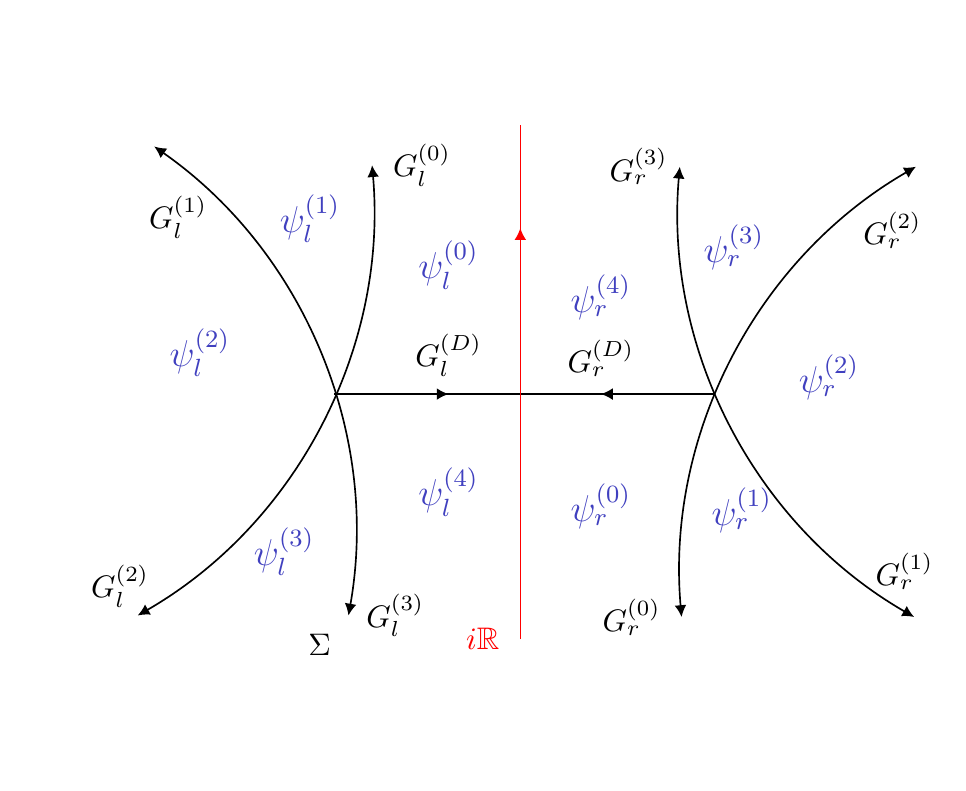}
\centering
\caption{Reducing the \pii{} \rh{} on to the imaginary axis.}
\end{figure}
\begin{remark}\label{remark:2}
The solution of the \rh{}:\ref{rhp:4} defines, via \eqref{cR} a solution of the \rh{\ \ref{rhp:5}}. Vice versa any solution of the \rh{\ \ref{rhp:5}} provides a solution to the \rh{}:\ref{rhp:4} by means of the inverse of the transformation \eqref{cR}. Thus we regard these two problems as {\it equivalent} in the sense that the solvability of one of them is necessary and sufficient condition for the solvability of the other. 
\end{remark}

For later use we compute the expression of the jump matrix $J$ in \eqref{J2}.
\begin{lemma}
The jump on the imaginary axis equals
\begin{gather} 
J(z,t) =  \Theta_{-}(z,t)^{-1} \Theta_{+} (z,t) =  \pr^{(0)}(z,t) \left[ \pl^{(4)}(z,t)\right]^{-1} = \left( \begin{array}{cc} \cA (z,t) & \cB (z,t) \\ \cC (z,t) & \cD(z,t) \end{array} \right) \label{jump_iR}
\end{gather}
where
\begin{gather}
\begin{array}{c}
\cA(z,t) = \frac{1}{h^4} \z^{\nu} \x^{\nu} e^{\frac{2i}{3} t} \left( - e^{-\pi i \nu} h^4 D_{-\nu}(i\z) D_{-\nu}(i\x) - \nu^2 e^{2\pi i \nu} D_{\nu-1}(\z) D_{\nu-1}(\x)  \right)  \\  \\
\cB(z,t) = -\frac{1}{ h^2} \left( \frac{z-z_{-}}{z-z_{+}} \right)^{2\nu} \z^{\nu} \x^{-\nu} \left( -i e^{-\pi i \nu} h^4 D_{-\nu}(i\z) D_{-\nu-1}(i\x) - \nu e^{2\pi i \nu} D_{\nu-1}(\z) D_{\nu}(\x)  \right) \\ \\
\cC(z,t) = \cB(-z,t) \, ; \quad \det J = 1. \\ \\
\end{array} \label{jump_entries}
\end{gather}
The variables $\zeta \equiv \zeta(z,t)$, $\xi \equiv \xi(z,t)$ are defined in \eqref{variablechange}, \eqref{varminus}; and $h$ is defined in terms of Stokes parameters in \eqref{Stokes}.
\end{lemma} 
\begin{proof}
Since $\Psi(z,t)$ has no jump on $i\bb{R}$, $J(z,t)$ can be determined solely in terms of $\pr^{(0)}(z,t)$ and $\pl^{(4)}(z,t)$. One can check the no monodromy condition at the origin,
\begin{gather}
\pr^{(0)}(z,t) \left[ \pl^{(4)}(z,t)\right]^{-1}=  \pr^{(4)}(z,t)\left[ \pl^{(0)}(z,t)\right]^{-1} .
\end{gather}
To ease the notation, we define
\begin{gather}
m(z) := \frac{z-z_{-}}{z-z_{+}},
\end{gather}
and observe that the following identities hold
\begin{gather}
\begin{array}{c}
\theta(z)= \frac{4}{3} z^3 -z = i\frac{\z^2}{4t}- \frac{1}{3}  = -i\frac{\x^2}{4t}+ \frac{1}{3}, \\ \\ (\z^2+\x^2) = -\frac{8 i t}{3}. 
\end{array}\label{rel}
\end{gather}
The function $\pr^{(0)}(z,t)$ is computed by substituting the zeroth sector solution of the parabolic cylinder function \eqref{Z_0} in \eqref{pr_sec},
\begin{gather}
\pr^{(0)}(z) =\left( \zeta(z) \frac{z-z_{-}}{z-z_{+}} \right)^{\nu\sigma_{3}} \left(-\frac{h_{1}}{s_{3}}  \right)^{-\sigma_{3}/2} e^{\frac{it}{3}\sigma_{3}} 2^{-\sigma_{3}/2} \left( \begin{array}{cc} \zeta(z) & 1 \\ 1 & 0 \end{array} \right) Z_{0}(\zeta(z)) \left(-\frac{h_{1}}{s_{3}}  \right)^{\sigma_{3}/2} \times e^{it\theta(z) \s} \nonumber \\ \nonumber \\
=\left[ \begin{array}{cc}   e^{i\pi \nu/2} e^{-\z^2/4} m(z)^{\nu} \z^{\nu} D_{-\nu}(i\z) & \frac{\nu}{h^2} e^{2it/3} e^{\z^2/4} m(z)^{\nu} \z^{\nu} D_{\nu-1}(\z) \\ i h^2 e^{i\pi \nu/2} e^{-2it/3} e^{-\z^2/4} m(z)^{-\nu} \z^{-\nu} D_{-\nu-1}(i\z) & e^{\z^2/4} m(z)^{-\nu} \z^{-\nu} D_{\nu}(\z)  \end{array} \right]. \label{phi_R0}
\end{gather}
The last line is obtained by using \eqref{rel},  and the following identity of parabolic cylinder functions
\begin{gather}  \frac{z}{2}D_{\nu}(z)+D'_{\nu}(z)=\nu D_{\nu-1}(z).
\end{gather}
The left parametrix $\pl^{(4)}$ can be obtained in a similar fashion, first by substituting $ Z_{4} = Z_{0}e^{-2i \pi \nu \s} $ from \eqref{h1} in \eqref{pr_sec} to obtain $\pr^{4}$ and using the relation \eqref{left_parametrix} to obtain $\pl^{(4)}$ as follows
\begin{gather}
\pl^{4} (z,t) = \sigma_{2} \pr^{4}(-z,t) \sigma_{2}  =\sigma_{2} \pr^{0}(-z,t) H_{D}^{-1} \sigma_{2} \nonumber \\ 
 = \left[ \begin{array}{cc}  e^{2 \pi i \nu} e^{\x^2/4} m(z)^{\nu} \x^{-\nu} D_{\nu}(\x) & -i h^2 e^{-3 \pi i \nu/2} e^{-2it/3} e^{\x^2/4} m(z)^{\nu} \x^{-\nu} D_{-\nu-1}(i\x) \\ - e^{2\pi i \nu}\nu h^{-2} e^{2it/3} e^{\x^2/4} m(z)^{-\nu} \x^{\nu} D_{\nu-1}(\x) & e^{-3\pi i \nu/2}e^{-y^2/4} m(z)^{-\nu} \x^{\nu} D_{-\nu}(i\x)  \end{array} \right]. \label{phi_L4}
\end{gather}
To obtain the last line, we substitute the expression for $Z_{0}$ \eqref{Z_0} and simplify the resulting expression using \eqref{rel}. Furthermore,
\begin{gather}
\det\left[ \pr^{(0)}(z,t) \right] =1 \, ; \quad \det\left[ \pl^{(4)}(z,t) \right] =1
\end{gather}
due to the following identity for the Wronskian  determinant of parabolic cylinder functions
\begin{gather}
\mathcal{W}\left[ D_{-\nu-1}(i\z), D_{\nu}(\z) \right] = ie^{-i\pi\nu/2}.
\end{gather}
The jump $J(z,t)$ is then obtained by a  straightforward substitution of \eqref{phi_R0} and \eqref{phi_L4} in \eqref{jump_iR}, and using \eqref{rel}.
\begin{gather}
J(z,t) = \pr^{(0)}(z,t) \left[ \pl^{(4)}(z,t)\right]^{-1} = \left( \begin{array}{cc} \cA (z,t) & \cB (z,t) \\ \cC (z,t) & \cD(z,t) \end{array} \right)
\end{gather}
where
\begin{gather}
\begin{array}{c}
\cA(z,t) =  \z^{\nu} \x^{\nu} e^{\frac{2i}{3} t} \left(  e^{-\pi i \nu}  D_{-\nu}(i\z) D_{-\nu}(i\x) + \nu^2 h^{-4} e^{2\pi i \nu} D_{\nu-1}(\z) D_{\nu-1}(\x)  \right)  \\  \\
\cB(z,t) =  \left(\frac{z-z_{-}}{z-z_{+}} \right)^{2\nu} \z^{\nu} \x^{-\nu} \left( i h^2 e^{- i \pi \nu}  D_{-\nu}(i\z) D_{-\nu-1}(i\x) + \nu h^{-2} e^{2\pi i \nu} D_{\nu-1}(\z) D_{\nu}(\x)  \right) \\ \\
\cC(z,t) =  \left( \frac{z-z_{-}}{z-z_{+}} \right)^{-2\nu} \z^{-\nu} \x^{\nu} \left( i h^{2} e^{-i \pi \nu}  D_{-\nu-1}(i\z) D_{-\nu}(i\x) + \nu h^{-2} e^{2\pi i \nu} D_{\nu}(\z) D_{\nu-1}(\x) \right) = \cB(-z,t) \\ \\
\cD(z,t) =   \z^{-\nu} \x^{-\nu} e^{-\frac{2i}{3} t} \left( - e^{-\pi i \nu} h^4 D_{-\nu-1}(i\z) D_{-\nu-1}(i\x) + e^{2\pi i \nu} D_{\nu}(\z) D_{\nu}(\x).  \right) 
\end{array}\label{def:JUMP_ABCD}
\end{gather}
It is obvious that $\det J(z,t) = 1$. Recall that \eqref{varminus}: $\x(-z,t) = \z(z,t)$ with $\z$ defined in \eqref{variablechange}, $h=\left(-\frac{h_1}{s_{3}}\right)^{1/2}$: \eqref{Stokes} where $h_{1}$, $\nu$ are determined by the Stokes parameters $s_{1}$, $s_{3}$ as in \eqref{h2}, \eqref{pr_par} respectively. One can further verify that the functions $\mathcal{A}, \mathcal{B}, \mathcal{C}, \mathcal{D}$ are analytic in a strip on the imaginary axis and
\begin{gather}
    \lim_{z\rightarrow \pm i \infty} J(z,t) = \mathbb{1}.
\end{gather}
\end{proof}


The two equivalent \rh{}s \ref{rhp:4} , \ref{rhp:5} give rise to two corresponding Malgrange forms. Although the two problems are equivalent, the two corresponding tau function may (and in fact do) differ, but only by a non-vanishing term which we now set up to compute.
Recalling the Malgrange form of \pii{} on $\Sigma$ in \eqref{def:1}:
\begin{equation}
d_{t} \log\tau_{_{\Sigma}} = \int_{\Sigma} \frac{dz}{2\pi i} \Tr\left[ \wP_{-}^{-1} \wP_{-}' \dot{G} G^{-1} \right]. \label{tau_s}
\end{equation}
Similarly on $i\bb{R}$, the \rh{} \ref{rhp:5} satisfies the jump condition $\Theta_{+} = \Theta_{-} J$ and the corresponding Malgrange form \eqref{DEF} is
\begin{equation}
d_{t} \log\tau_{_{i\bb{R}}} = \int_{i\bb{R}} \frac{dz}{2\pi i} \Tr\left[  \Theta_{-}^{-1} \Theta_{-}' \dot{J} J^{-1}   \right]. \label{tau_iR}
\end{equation}
\begin{proposition}\label{proposition:1}
The Malgrange forms corresponding to the \rh{}s on the contours $\Sigma$ and $i\bb{R}$ are related as 
\begin{gather}
d_{t} \log\tau_{_{\Sigma}} = d_{t} \log\tau_{_{i\bb{R}}} - \int_{i\bb{R}}\frac{dz}{2\pi i}  \mathcal{F}(z,t;\nu,h) -\left[ \frac{4i\nu}{3} + \frac{ \nu^2}{t}  \right]-2g\left(\nu, h; t\right),
 \end{gather}
where $\mathcal{F}(z,t;\nu,h)$ is a regular function explicit in terms of parabolic cylinder functions, see \eqref{F_def}.
\end{proposition}

\begin{proof}\footnote[4]{In the proof, we drop the $z$, $t$ dependence for the ease of writing. All the functions here on depend on $z$, $t$ unless specified. One should pay attention to the total derivative w.r.t $t$ as it involves a derivative w.r.t $z$ due to the transformation \eqref{cov}. We thank Nikolai Iorgov and Yuri Zhuravlev for bringing this subtlety with the $t$ derivative to our attention}
We begin by computing the following trace on $i\mathbb{R}$
\begin{gather}\label{Tr_RJ}
\Tr \left\lbrace  \Theta_{-}^{-1} \Theta_{-}' \dot{J} J^{-1}  \right\rbrace. 
\end{gather}
Computing \eqref{Tr_RJ} term by term using \eqref{cR}: $\Theta_{-} = \wP \pri $,
\begin{gather}  
\Theta_{-}^{-1} \Theta_{-}' = (\wP \pri)^{-1} (\wP \pri)' 
= \pr \wP^{-1} \left( \wP' \pri - \wP \pri \pr' \pri  \right) \nonumber \\
= \pr \left( \wP^{-1} \wP' - \pri \pr'  \right) \pri . \label{4.11}
\end{gather}
Since \eqref{jump_iR}: $J= \pr \pli $,
\begin{gather}
\dot{J} J^{-1} = \frac{\partial}{\partial t}{(\pr \pli)} (\pr \pli)^{-1} 
= \left(  \dot{\pr} \pli - \pr \pli \dot{\pl} \pli  \right) \pl \pri \nonumber \\
= -\pr \Delta \left( \p^{-1} \dot{\p}  \right) \pri, \label{4.12}
\end{gather}
where 
\begin{equation}
\Delta (\p^{-1} \dot{\p} ) = \pli \dot{\pl} -\pri \dot{\pr}. \nonumber
\end{equation}
Substituting \eqref{4.11} and \eqref{4.12} in \eqref{Tr_RJ} and using cyclicity of trace,
\begin{gather}
\Tr \left\lbrace  \Theta_{-}^{-1} \Theta_{-}' \dot{J} J^{-1}  \right\rbrace \ = \Tr \left\lbrace  \left( - \wP^{-1} \wP' + \pri \pr'   \right) \Delta \left(\p^{-1} \dot{\p}   \right)  \right\rbrace. \label{4.13} 
\end{gather}
Since the term $\pri \pr'   \Delta \left( \p^{-1} \dot{\p} \right) $ is integrated on $i\bb{R}$ in \eqref{tau_iR}, 
{\small \begin{gather}
\int_{i\bb{R}} \frac{dz}{2\pi i }\Tr \left[ \pr^{-1} \pr'  \Delta \left( \p^{-1}\dot{\p}  \right)  \right] =  \int_{i\bb{R}} \frac{dz}{2\pi i} \Tr\left[ \left(\pr^{(0)}\right)^{-1} \left(\pr^{(0)}\right)'  \left\lbrace \left(\pl^{(4)}\right)^{-1} \dot{\pl}^{(4)}  - \left(\pr^{(0)}\right)^{-1} \dot{\pr}^{(0)}   \right\rbrace  \right]
\end{gather}}
with $\pr^{(0)}$ defined in \eqref{phi_R0}, $\pl^{(4)}$ in \eqref{phi_L4}. We collect the explicit terms and compute them in the end. Since $\wP$ has no jump on $i\mathbb{R}$, using Cauchy theorem
\begin{gather}
-\int_{i\mathbb{R}} \frac{d z}{2\pi i} \Tr  \left\lbrace  \wP^{-1} \wP' \Delta \left(  \p^{-1} \dot{\p}  \right) \right\rbrace = -\int_{i\mathbb{R}} \frac{d z}{2\pi i} \Tr \Delta  \left(  \wP^{-1} \wP' \left( \p^{-1} \dot{\p}   \right) \right) \nonumber \\
 = \int_{\Sigma} \frac{dz}{2\pi i} \Tr \Delta  \left( \wP^{-1} \wP' \left( \p^{-1} \dot{\p}  \right) \right) \nonumber \\
 = \int_{\Sigma_L} \frac{dz}{2\pi i} \Tr \Delta  \left(  \wP^{-1} \wP' \left(  \p^{-1} \dot{\p}   \right) \right) + \int_{\Sigma_R} \frac{dz}{2\pi i}  \Tr \Delta  \left( \wP^{-1} \wP' \left(  \p^{-1} \dot{\p}  \right) \right), \label{4.14} 
\end{gather}
where $\Sigma_{L,R}$ are $\Sigma$ restricted to the left and right half-planes respectively. Since $\Psi$ has jumps on $\Sigma_{L}$,
\begin{gather}
 \int_{\Sigma_{L}} \frac{dz}{2\pi i}\Tr \Delta\left( \wP^{-1} \wP' \left(  \p^{-1} \dot{\p}   \right) \right) = \int_{\Sigma_{L}} \frac{dz}{2\pi i}\Tr \left\lbrace \wP_{+}^{-1} \wP_{+}' \left( \plp^{-1} \dot{\plp}  \right) - \wP_{-}^{-1} \wP_{-}' \left( \plm^{-1} \dot{\plm}   \right) \right\rbrace.  \label{4.15}
 \end{gather}
similarly on $\Sigma_{R}$
\begin{gather}
\int_{\Sigma_{R}} \frac{dz}{2\pi i} \Tr \Delta\left(  \wP^{-1} \wP' \left(  \p^{-1} \dot{\p}   \right) \right) = \int_{\Sigma_{R}} \frac{dz}{2\pi i} \Tr \left\lbrace  \wP_{+}^{-1} \wP_{+}' \left( \prp^{-1} \dot{\prp}  \right) - \wP_{-}^{-1} \wP_{-}' \left( \prm^{-1} \dot{\prm}   \right) \right\rbrace.
\end{gather}
In order to estimate \eqref{4.14}, we begin by computing the integrand on $\Sigma_L$. Computing \eqref{4.15} term by term using \eqref{J_Gamma}: $\Psi_{+} =\Psi_{-} G_{l}$,
\begin{gather}
 \wP_{+}^{-1} \wP_{+}' = (\wP_{-} G_{l})^{-1} (\wP_{-} G_{l})' 
 = G_{l}^{-1} \wP_{-}^{-1} (\wP_{-}' G_{l} + \wP_{-} G_{l}' ) \nonumber \\
 = G_{l}^{-1} (\wP_{-}^{-1} \wP_{-}'  +  G_{l}' G_{l}^{-1} )G_{l}. \label{4.16}
\end{gather}
Since \eqref{G_rl}: $\plp =\plm  G_{l}$,
\begin{gather}
\plp^{-1} \dot{\plp}  = G_{l}^{-1} \plm^{-1} \left( \dot{\plm}  G_{l} + \plm  \dot{G}_{l} \right)  \nonumber \\
= G_{l}^{-1}  \left( \plm^{-1} \dot{\plm}   +  \dot{G}_{l} G_{l}^{-1}     \right) G_{l}. \label{4.17}
\end{gather}
The product of \eqref{4.16} and \eqref{4.17} under the trace reads
\begin{gather}
\Tr \left\lbrace  \wP_{+}^{-1} \wP_{+}' \left( \dot{\plp} \plp^{-1}  \right) \right\rbrace = \Tr \left[ \left(\wP_{-}^{-1} \wP_{-}'  +  G_{l}' G_{l}^{-1} \right) \left( \plm^{-1} \dot{\plm}   +  \dot{G}_{l} G_{l}^{-1}     \right)  \right]. \label{4.18}
\end{gather}
Substituting \eqref{4.18} in \eqref{4.15},
\begin{gather}
\Tr \Delta\left(  \wP^{-1} \wP' \left( \dot{\p} \p^{-1}  \right) \right) = \Tr\left[ \wP_{-}^{-1} \wP_{-}' \dot{G}_{l} G_{l}^{-1} + G_{l}' G_{l}^{-1}  \left( \plm^{-1} \dot{\plm}   +  \dot{G}_{l} G_{l}^{-1}  \right) \right] \label{4.19}
\end{gather}
A parallel computation for $\Sigma_R$ gives 
\begin{gather}
\Tr \Delta\left(  \wP^{-1} \wP' \left( \dot{\p} \p^{-1}  \right) \right) = \Tr\left[ \wP_{-}^{-1} \wP_{-}' \dot{G}_{r} G_{r}^{-1} + G_{r}' G_{r}^{-1}  \left( \prm^{-1} \dot{\prm}   +  \dot{G}_{r} G_{r}^{-1}     \right) \right]. \label{4.19R}
\end{gather}
 Summing the terms \eqref{4.19} and \eqref{4.19R}, we obtain that
\begin{gather}
d_{t} \log\tau_{_{i\bb{R}}} = \int_{i\bb{R}} \frac{dz}{2\pi i} \Tr\left[\Theta_{-}^{-1} \Theta_{-}' \dot{J}J^{-1}\right] = \int_{\Sigma_{R}} \frac{dz}{2\pi i} \Tr \left[ \wP_{-}^{-1} \wP_{-}' \dot{G}_{r} G_{r}^{-1} \right] + \int_{\Sigma_{L}}\frac{dz}{2\pi i} \Tr\left[ \wP_{-}^{-1} \wP_{-}' \dot{G}_{l} G_{l}^{-1} \right] \nonumber \\ 
 + \int_{\Sigma_{L}}\frac{dz}{2\pi i} \Tr \left[ G_{l}' G_{l}^{-1}  \left( \plm^{-1} \dot{\plm}   +  \dot{G}_{l} G_{l}^{-1}     \right)   \right] + \int_{\Sigma_{R}} \frac{dz}{2\pi i} \Tr\left[ G_{r}' G_{r}^{-1}  \left( \prm^{-1} \dot{\prm}   +  \dot{G}_{r} G_{r}^{-1}     \right) \right] \nonumber \\
 + \int_{i\bb{R}} \frac{dz}{2\pi i} \Tr\left[  \pr^{-1} \pr' \Delta \left( \p^{-1}\dot{\p}  \right)  \right] \nonumber \\
 = \int_{\Sigma} \frac{dz}{2\pi i} \Tr \left[ \wP_{-}^{-1} \wP_{-}' \dot{G} G^{-1} \right] + \int_{i\bb{R}} \frac{dz}{2\pi i} \Tr\left[ \pr^{-1} \pr'  \Delta \left( \p^{-1}\dot{\p}  \right)  \right] \nonumber \\ + \int_{\Sigma_{L}} \frac{dz}{2\pi i} \frac{dz}{2\pi i} \Tr \left[ G_{l}' G_{l}^{-1}  \left( \plm^{-1} \dot{\plm}   +  \dot{G}_{l} G_{l}^{-1}     \right)   \right] + \int_{\Sigma_{R}} \frac{dz}{2\pi i} \Tr\left[ G_{r}' G_{r}^{-1}  \left( \prm^{-1} \dot{\prm}   +  \dot{G}_{r} G_{r}^{-1}     \right) \right] 
  \nonumber \\
= d_{t} \log\tau_{_{\Sigma}} + \int_{i\bb{R}} \frac{dz}{2\pi i} \Tr\left[ \pr^{-1} \pr'  \Delta \left( \p^{-1}\dot{\p}  \right)  \right] \nonumber \\ + \int_{\Sigma_{L}} \frac{dz}{2\pi i} \Tr \left[ G_{l}' G_{l}^{-1}  \left( \plm^{-1} \dot{\plm}   +  \dot{G}_{l} G_{l}^{-1}     \right)   \right] + \int_{\Sigma_{R}} \frac{dz}{2\pi i} \Tr\left[ G_{r}' G_{r}^{-1}  \left( \prm^{-1} \dot{\prm}   +  \dot{G}_{r} G_{r}^{-1}     \right) \right].
\end{gather}
Notice that $\p_{r,l}$ and $G_{r,l}$ are completely determined in terms of parabolic cylinder functions. The final expression is
\begin{gather}
d_{t} \log\tau_{_{\Sigma}} = d_{t} \log\tau_{_{i\bb{R}}} - \int_{i\bb{R}}\frac{dz}{2\pi i} \Tr\left[ \pr^{-1} \pr'  \Delta \left( \p^{-1}\dot{\p}  \right)  \right] \nonumber \\ - \int_{\Sigma_{L}} \frac{dz}{2\pi i} \Tr \left[ G_{l}' G_{l}^{-1}  \left( \plm^{-1} \dot{\plm}   +  \dot{G}_{l} G_{l}^{-1}     \right)   \right] - \int_{\Sigma_{R}} \frac{dz}{2\pi i} \Tr\left[ G_{r}' G_{r}^{-1}  \left( \prm^{-1} \dot{\prm}   +  \dot{G}_{r} G_{r}^{-1}     \right) \right].  \label{explicit_phi}
\end{gather}

The following can be said about the explicit terms in \eqref{explicit_phi}.
\begin{itemize}
 \item We can completely determine the integrals on $\Sigma_{R,L}$. The symmetry relations \eqref{left_parametrix}, \eqref{G_symm}
 imply that
 \begin{gather}
 \int_{\Sigma_{L}} \frac{dz}{2\pi i} \Tr \left[ G_{l}' G_{l}^{-1}  \left( \plm^{-1} \dot{\plm}   +  \dot{G}_{l} G_{l}^{-1}     \right)   \right] = \int_{\Sigma_{R}} \frac{dz}{2\pi i} \Tr \left[ G_{r}' G_{r}^{-1}  \left( \prm^{-1} \dot{\prm}   +  \dot{G}_{r} G_{r}^{-1}     \right)   \right]. \label{sym}
 \end{gather}
Furthermore,  \eqref{h2} implies that the jump $G_{r}^{(k)}$ in \eqref{GH} is lower triangular for $k = 0,2$; upper triangular for $k = 1,3$; diagonal and constant for $k=4$. Therefore,
\begin{gather}
\Tr\left[ G_{r}' G_{r}^{-1}  \dot{G}_{r} G_{r}^{-1} \right]= \Tr\left[ G_{l}' G_{l}^{-1}  \dot{G}_{l} G_{l}^{-1} \right]=0.
\end{gather}

We now proceed to compute the following term in \eqref{explicit_phi}
\begin{gather}
\int_{\Sigma_{R}} \frac{dz}{2\pi i} \Tr \left[ G_{r}' G_{r}^{-1}  \prm^{-1} \dot{\prm}    \right] = \sum_{k=1}^5\int_{\Sigma_{k}} \frac{dz}{2\pi i}  \Tr \left[ (G_{r}^{(k)})' (G_{r}^{(k)})^{-1}(\prm^{(k-1)})^{-1}  \dot{\prm}^{(k-1)} \right]. \label{4.24}
\end{gather}
In each sector, $\pr$ and $G_{r}$ can be computed starting from $\pr^{(0)}$ in \eqref{phi_R0}, and the jumps in \eqref{GH}. A lengthy but straightforward computation yields 
\begin{gather}
 \int_{\Sigma_{R}}\frac{dz}{2\pi i} \Tr\left[  G_{r}' G_{r}^{-1}  \left( \prm^{-1} \dot{\prm}   +  \dot{G}_{r} G_{r}^{-1}     \right)   \right]  = \left[ \frac{2i\nu}{3} + \frac{\nu^2}{2t}  \right]+g\left(\nu, h; t\right),
\end{gather}
where $g(.)$ is a function depending on $t$ with the following behaviour for $t\to -\infty$
\begin{align*}
    g(\nu,h;t)=-\frac{2 i \left(3 \nu ^4+33 \nu ^2-1\right) \nu }{81 t^2}+\mathcal{O}((-t)^{-3}).
\end{align*}
The relation \eqref{sym} then implies,
{\small \begin{align}
 \int_{\Sigma_{L}} \frac{dz}{2\pi i}\Tr \left[ G_{l}' G_{l}^{-1}  \left( \plm^{-1} \dot{\plm}   +  \dot{G}_{l} G_{l}^{-1}     \right)   \right] &+ \int_{\Sigma_{R}} \frac{dz}{2\pi i} \Tr\left[ G_{r}' G_{r}^{-1}  \left( \prm^{-1} \dot{\prm}   +  \dot{G}_{r} G_{r}^{-1}     \right) \right] \nonumber\\
 &= \left[ \frac{4i\nu}{3} + \frac{\nu^2}{t}  \right]+ 2g\left(\nu, h; t\right).
\end{align}}

\item The remaining explicit term in \eqref{explicit_phi}
{\small \begin{gather}
\int_{i\bb{R}} \frac{dz}{2\pi i} \Tr\left[ \pr^{-1} \pr'  \Delta \left( \p^{-1}\dot{\p}  \right)  \right] \equiv \int_{i\bb{R}} \frac{dz}{2\pi i} \Tr\left[ \left(\pr^{(0)}\right)^{-1} \left(\pr^{(0)}\right)'  \left\lbrace \left(\pl^{(4)}\right)^{-1} \dot{\pl}^{(4)}  - \left(\pr^{(0)}\right)^{-1} \dot{\pr}^{(0)}   \right\rbrace  \right]. \label{exp} 
\end{gather}}
The functions $\pr^{(0)}$ and $\pl^{(4)}$ depend on $z$ through $\zeta(z,t)$ as in \eqref{variablechange} and $\x (z,t)$ as in \eqref{varminus} respectively. In order to solve the integral, we need to compute integrals of the form
\begin{equation}
\int \frac{dz}{2\pi i} D_{\nu}(\z) D_{\mu}(\x) D_{-\rho}(i \z) D_{-\sigma}(i\x),
\end{equation}
which is not exactly solvable. The expression \eqref{exp} is however, explicit. Defining a function $\mathcal{F}$ as
\begin{gather}
  \mathcal{F}(z,t;\nu,h) :=\Tr\left[\pr^{-1} \pr' \left( \pl^{-1}\dot{\pl}  - \pr^{-1}  \dot{\pr}  \right)  \right], \label{F_def}
\end{gather}
The final expression in \eqref{explicit_phi} reads
\begin{gather}
d_{t} \log\tau_{_{\Sigma}} = d_{t} \log\tau_{_{i\bb{R}}} - \int_{i\bb{R}} \frac{dz}{2\pi i}  \mathcal{F}(z, t;\nu,h) -\left[ \frac{4i\nu}{3} + \frac{ \nu^2}{t}  \right]-2g\left(\nu, h; t\right). \label{final1}
\end{gather}
\end{itemize}
\end{proof}

\section{Integrable kernel and Fredholm determinant}\label{section:4}

Up to this point, we started with the \rh{} of \pii{} in fig. \ref{fig:2}, used the description of the local parametrices in terms of parabolic cylinder functions in the subsection \ref{subsection:2.1} to define a \rh{} on $i\bb{R}$ \eqref{J2} in section \ref{section:3}. We then showed that the corresponding Malgrange forms are related in proposition \ref{proposition:1}. Our goal now
reduces to expressing $\tau_{_{i\bb{R}}}$ as a Fredholm determinant. 

It is known that a jump $J(z,t)\in SL(2, \bb{C})$ on non-intersecting contours can be expressed in terms of lower and upper triangular matrices called the LULU decomposition and the corresponding \tfn{} can then be written as a Fredholm determinant of an integrable operator \cite{MalB}. Here, we modify the construction in \cite{MalB} by using LDU decomposition\footnotetext[5]{The author thanks A.Its for suggesting LDU decomposition.} instead, which then gives us a simpler kernel. In this section, we
\begin{itemize}
\item[1.] transform RHP:\ref{rhp:5} on to a set of two parallel lines with lower and upper triangular jumps using the LDU decomposition,
\item[2.] formulate the \tfn{} on the set of parallel lines, call it $\tau_{_{LU}}$ as a Fredholm determinant of an integrable operator, and
\item[3.] prove that the Malgrange forms on the contours LU and $i \bb{R}$ coincide. 
\end{itemize}
\subsection{LU decomposition} \label{subsection:4.1}

The \rh{} on $i\bb{R}$ can be transformed on to a set of two parallel lines with jumps that are upper and lower triangular respectively. 
We decompose the jump $J(z,t)$ \eqref{jump_iR} into lower, diagonal and upper triangular matrices, called the LDU decomposition \cite{RHPdeift}, which recasts the \rh{}:\ref{rhp:5} on to a set of three parallel lines.
\begin{gather}
J(z,t)= \left( \begin{array}{cc}
\cA(z,t) & \cB(z,t) \\ \cC(z,t) & \cD(z,t)
\end{array} \right) = \left( \begin{array}{cc}
1 & 0 \\ \frac{\cC(z,t)}{\cA(z,t)} & 1
\end{array} \right)\left( \begin{array}{cc}
\cA(z,t) & 0 \\ 0 & \frac{1}{\cA(z,t)}
\end{array} \right) \left( \begin{array}{cc}
1 & \frac{\cB(z,t)}{\cA(z,t)} \\ 0 & 1
\end{array} \right) \nonumber\\
 := F_{1}(z,t) F_{2}(z,t) F_{3}(z,t). \label{LDU}
\end{gather}
\begin{figure}[H] 
\includegraphics[width=8cm]{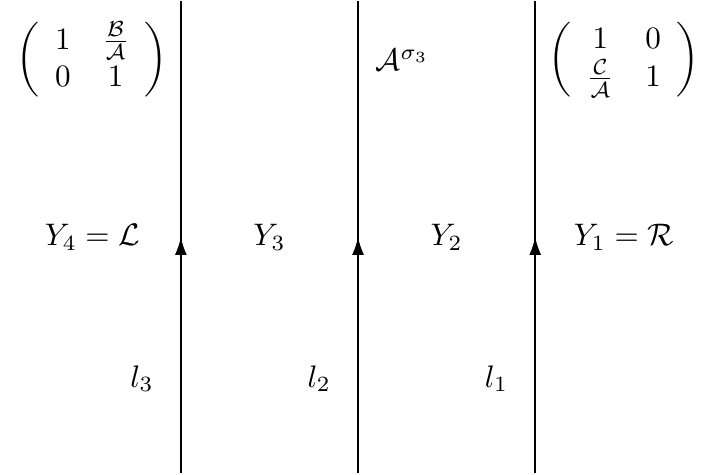}
\centering
\caption{LDU decomposition \label{fig:5}}
\end{figure}
The function $Y(z,t)$ then solves the following \rh{}.
\begin{rhp} \label{rhp:6}
\begin{itemize}
  \item $Y(z,t)$ is a piecewise analytic in $\bb{C}\backslash \left( \cup_{i=1}^{3} l_{i} \right)$.
\item On each line $l_i$ in fig. \ref{fig:5}, the following jump condition holds
\begin{gather}
Y_{i+1}(z,t) = Y_{i}(z,t) F_{i}(z,t),\label{F_i}
\end{gather}
with the identification
\begin{gather}
Y_{4}(z,t) = \Theta_{+}(z,t) \, ; \quad Y_{1}(z,t) = \Theta_{-}(z,t). \label{ide}
\end{gather}
$\Theta_{\pm}$ are defined in \eqref{cR}.
\item $\lim_{z\rightarrow \infty} Y(z,t) = \bb{1}$.
\end{itemize}
\end{rhp}
The \rh{}:\ref{rhp:6} can be further transformed with the observation that the function $\varphi(z,t)^{\s}$ defined as
\begin{gather}
\varphi (z,t) := \exp\left[\int_{i\bb{R}} \frac{dw}{2\pi i} \frac{\log \cA(w,t)}{z-w}\right], \label{varphi}
\end{gather}
solves \rh{} on $l_{2}$ with the diagonal jump $\cA^{\s}$ locally with $\cA$ defined in \eqref{jump_entries}.
The ratio of $Y(z,t)$, $\varphi(z,t)^{\s}$ 
\begin{gather}
\wY_{i}(z,t) :=Y_{i}(z,t) \varphi(z,t)^{-\sigma_{3}} \label{wY}
\end{gather} 
is such that $\wY (z,t)$ jumps only on $l_{1} \cup l_{3}$ and solves the following \rh{}.
\begin{rhp}\label{rhp:8}
\begin{itemize}
\item 
$\wY(z,t)$ is piece-wise analytic in $\bb{C}\backslash(l_{1}\cup l_{3})$.
\begin{figure}[H] 
\includegraphics[width=8cm]{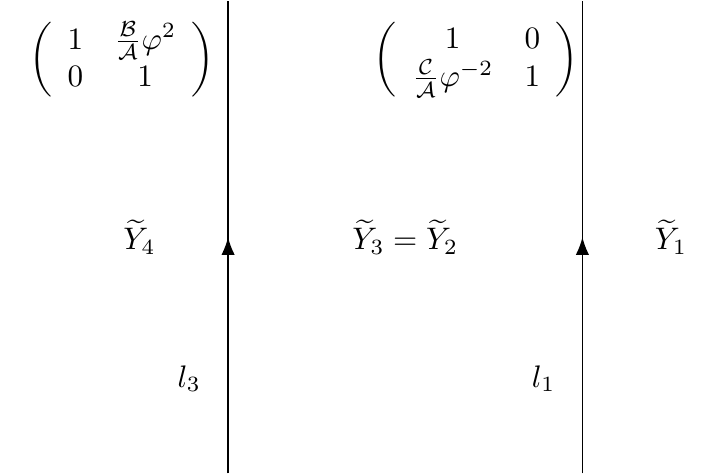}
\centering
\caption{\rh{} with lower and upper triangular jumps. \label{fig:4}}
\end{figure}
\item The following jump conditions are valid on the contours $l_{i}$, $i=1,3$
\begin{gather}
\wY_{i+1}(z,t) = \wY_{i}(z,t) \wF_{i}(z,t) \label{tF_1}
\end{gather}
where 
\begin{gather} \label{wF}
\wF_{1}(z,t) = \left( \begin{array}{cc}
1 & 0 \\ \frac{\cC(z,t)}{\cA(z,t)} \varphi(z,t)^{2} & 1
\end{array} \right)\, ; \quad 
 \wF_{3}(z,t) = \left( \begin{array}{cc}
1 & \frac{\cB(z,t)}{\cA(z,t)}\varphi(z,t)^{-2} \ \\ 0 & 1
\end{array} \right).
\end{gather}
Note that $\wY(z,t)$ has no jump on $l_{2}$, implying that $\wY_{3}(z,t) = \wY_{2}(z,t)$. 
\end{itemize}
\end{rhp}
The \rh{}:\ref{rhp:8} in fig. \ref{fig:4} is of the 'integrable' type and its solvability is determined by the invertibility of an integrable operator i.e, its \tfn{} is the Fredholm determinant of an integrable operator.



\subsection{Integrable kernel}\label{subsection:4.2}
\begin{proposition}\label{proposition:2}
The \tfn{} on $l_{1}\cup l_{3}$ denoted by $\tau_{_{LU}}$ is a Fredholm determinant of an integrable operator 
\begin{gather}
\tau_{_{LU}} = \det \left[ \bb{1}_{_{L^2(i\bb{R})}} - \widetilde{\mathcal{K}} \right]
\end{gather}
where 
\begin{gather}
\left(\widetilde{\mathcal{K}}h\right)(z) = \frac{\cC(z,t)}{\cA(z,t)}\int_{i\bb{R}}\frac{dw}{2\pi i } \int_{i\bb{R}+ \epsilon} \frac{d\widetilde{w}}{2\pi i} \frac{\varphi_{+}^{2}(w)  \varphi_{+}^{-2}(\widetilde{w})}{(z-\widetilde{w})(\widetilde{w}-w)} \cA(\widetilde{w},t) \cB(\widetilde{w},t) h(\widetilde{w}).
\end{gather}
The functions $\cA$, $\cB$, $\cC$ are defined in \eqref{jump_entries} and $\varphi_{+}$ is the positive (left of the imaginary axis) boundary value of \eqref{varphi}.
\end{proposition}
\begin{proof}
Let us recall the jumps in \eqref{wF}
\begin{gather}
\wF(z,t) = \begin{cases} 
\wF_{1}(z,t) = \left( \begin{array}{cc}
1 & 0 \\ \frac{\cC(z,t)}{\cA(z,t)} \varphi(z,t)^{2} & 1
\end{array} \right)\, ; \textrm{on}\, l_{1} \\
 \wF_{3}(z,t) = \left( \begin{array}{cc}
1 & \frac{\cB(z,t)}{\cA(z,t)}\varphi(z,t)^{-2} \ \\ 0 & 1
\end{array} \right)\, ; \textrm{on}\, l_{3} \end{cases}. \label{tF}
\end{gather}
We define the functions
\begin{gather}
f(z,t) = \frac{1}{2\pi i} \left( \begin{array}{c}
\frac{\cB(z,t)}{\cA(z,t)} \chi_{3}(z) \\  \\
\frac{\cC(z,t)}{\cA(z,t)}\chi_{1}(z)
\end{array} \right) \quad ; \quad g(z,t) = \left( \begin{array}{c}
 \varphi(z,t)^2 \chi_{1}(z) \\ \\
 \varphi(z,t)^{-2} \chi_{3}(z)
\end{array} \right) \label{fg}
\end{gather}
where $\chi_{1}(z)$, $\chi_{3}(z)$ denote the characteristic functions on the contours $l_{1}$, $l_{3}$ respectively. The jump $\wF(z,t)$ can be written in terms of \eqref{fg} as
\begin{equation}
\wF = 1 - 2\pi i f(z) g^{T}(z),
\end{equation}
and clearly $f^{T}(z) g(z) =0$. The associated integrable kernel is then 
\begin{gather}
K(z,w) = \frac{f^{T}(z) g(w)}{z-w} = \nonumber \\
= \frac{1}{2\pi i (z-w)}\left(\begin{array}{cc}
\chi_{1}(z) & \chi_{3}(z)
\end{array}  \right) \left(\begin{array}{ccc}
0 & & \frac{\cC(z,t)}{\cA(z,t)}\varphi^{-2}(w,t) \\ \frac{\cB(z,t)}{\cA(z,t)}\varphi^{2}(w,t) & & 0
\end{array}  \right) \left(\begin{array}{c}
\chi_{1}(w) \\ \chi_{3}(w)
\end{array}  \right) \nonumber \\
\equiv \left(\begin{array}{cc}
\chi_{1}(z) & \chi_{3}(z)
\end{array}  \right) \left(\begin{array}{ccc}
0 & & K_{31}(z,w) \\ K_{13}(z,w) & & 0
\end{array}  \right) \left(\begin{array}{c}
\chi_{1}(w) \\ \chi_{3}(w)
\end{array}  \right). \label{kernel}
\end{gather} 
The kernels $ K_{13}(z,w)$ and $ K_{31}(z,w)$ in \eqref{kernel} take the form 
\[
 K_{13}(z,w)=\dfrac{\cB(z,t)\varphi^{2}(w,t)}{(2\pi i) \cA(z,t) (z-w)}
 \]
 \[
  K_{31}(z,w)= \dfrac{\cC(z,t)\varphi^{-2}(w,t)}{(2\pi i) \cA(z,t) (z-w)}\,.
  \]
We introduce the operators
\begin{gather}
\begin{array}{c}
{\mathcal K}_{31} : L^{2}(l_{3}) \rightarrow L^{2}(l_{1}) \\  \\
{\mathcal K}_{13} : L^{2}(l_{1}) \rightarrow L^{2}(l_{3}),
\end{array}
\end{gather}
defined as 
\begin{gather}
\begin{array}{c}
\left({\mathcal K}_{31} h\right)(z) = \int_{l_{3}} K_{31}(z,w)h(w)dw, \\
\\
\left({\mathcal K}_{13} \widetilde{h}\right)(z) = \int_{l_{1}}K_{13}(z,w)\widetilde{h}(w)dw. 
\end{array} \label{component}
\end{gather}
The \tfn{} corresponding to \rh{}: \ref{rhp:8} is then 
\begin{gather}
\tau_{_{LU}}(t) = \det\left[\bb{1}_{_{L^2(l_{1}\cup l_{3})}} - 
\begin{pmatrix} 0&{\mathcal K}_{31}\\
{\mathcal K}_{13}&0\end{pmatrix} \right]. \label{Det}
\end{gather}

Since $\varphi^2(w,t)$ is analytic in $\Re (w)>0$ and $\lim_{w\rightarrow \infty} \varphi(w,t)= 1$, $\mathcal{K}_{13}$, $\mathcal{K}_{31}$ are Trace--class.
Therefore we can write $\tau_{_{LU}}(t) $ in the form 
\begin{gather}
\tau_{_{LU}}(t) =  \det\left[\bb{1}_{_{L^2(l_{3})}} - {\mathcal K}_{13} \circ {\mathcal K}_{31}  \right]. \label{tc}
\end{gather} 

The form of the \tfn{} \eqref{tc} can be further modified such that the operator acts on $L^{2}(i\bb{R})$ instead of $L^{2}(l_{3})$. We begin by splitting the function $h(z)$ as
\begin{gather}
h(z) = h_{L}(z) + h_{R}(z)
\end{gather}
where $h_{L,R}(z)$ are analytic to the left and right of $l_{3}$ respectively, and $h_{L,R}(z) = \mathcal{O}(z^{-1})$ as $z\rightarrow \infty$. The integrable operator \eqref{component} acts on $h(z)$ as
\begin{gather}
\left( {\mathcal K}_{13} {\mathcal K}_{31}h_{R}\right)(z) \equiv 0  \Rightarrow \left( {\mathcal K}_{13} {\mathcal K}_{31} h\right)(z) = \left({\mathcal K}_{13}{\mathcal  K}_{31} h_{L}\right)(z).
\end{gather}
We can therefore move the integration in $w$ from $l_{3}$ to $i\bb{R}$ in \eqref{component} and identify the space of functions $\left({\mathcal K}_{31} h\right)(z)$ with $H_{R}(i\bb{R})$, the Hardy space on the right half--plane. So, the operator
\begin{gather}
\left(\widetilde{\mathcal{K}}h\right)(z) :=\left({\mathcal K}_{13} {\mathcal K}_{31} h \right)(z) =  \frac{\cC(z,t)}{\cA(z,t)} \int_{l_{1}} \frac{d\widetilde{w}}{2\pi i} \int_{i\bb{R}} \frac{dw}{2\pi i} \frac{\varphi^{-2}(\widetilde{w})}{z-\widetilde{w}} \frac{\cB(\widetilde{w},t)}{\cA(\widetilde{w},t)} \frac{\varphi_{+}^{2}(w,t)}{\widetilde{w}-w} h(w)
\end{gather}
 The kernel, $\widetilde{K}(z,w)$ is 
\begin{gather}
\widetilde{K}(z,w) = \frac{\cC(z,t)}{\cA(z,t)} \varphi_{+}^{2}(w) \int_{l_{1}} \frac{d\widetilde{w}}{2\pi i} \frac{\varphi^{-2}(\widetilde{w})}{(z-\widetilde{w})(\widetilde{w}-w)} \frac{\cB(\widetilde{w},t)}{\cA(\widetilde{w},t)}.
\end{gather}
We can now move $l_{1}$ to $i\bb{R}+\epsilon$ from the right without changing the kernel $\widetilde{\mathcal{K}}$
\begin{align}
\widetilde{K}(z,w)& = \frac{\cC(z,t)}{\cA(z,t)} \varphi_{+}^{2}(w) \int_{i\bb{R}+\epsilon} \frac{d\widetilde{w}}{2\pi i} \frac{\varphi_{-}^{-2}(\widetilde{w})}{(z-\widetilde{w})(\widetilde{w}-w)} \frac{\cB(\widetilde{w},t)}{\cA(\widetilde{w},t)}\\
&=  \frac{\cC(z,t)}{\cA(z,t)} \varphi_{+}^{2}(w) \int_{i\bb{R}+ \epsilon} \frac{d\widetilde{w}}{2\pi i} \frac{\varphi_{+}^{-2}(\widetilde{w})}{(z-\widetilde{w})(\widetilde{w}-w)} \cA(\widetilde{w},t) \cB(\widetilde{w},t),
\end{align}
where in the last identity we use the relation from \eqref{F_i}, \eqref{varphi}: $\varphi_{+}(\widetilde{w}) = \varphi_{-}(\widetilde{w})\cA(\widetilde{w},t) $.
Therefore we conclude from \eqref{tc} and the above discussion that 
\begin{gather}
\tau_{_{LU}}(t) =  \det\left[\bb{1}_{_{L^2(l_{3})}} - {\mathcal K}_{13} \circ {\mathcal K}_{31}  \right]=\det\left[\bb{1}_{_{L^2(i\bb{R})}} - \widetilde{\mathcal{K}}  \right]\label{result}
\end{gather}
\end{proof}

\subsection{Malgrange forms}\label{subsection:4.3}
In \eqref{result}, we expressed the \tfn{} on LU as a Fredholm determinant. To relate $\tau_{_{LU}}$ to $\tau_{_{\Sigma}}$ in \eqref{tau_s}, we will first prove that the \tfn{} corresponding to the \rh{}:\ref{rhp:6}, call it $\tau_{_{LDU}}$, is equal to $\tau_{_{i\bb{R}}}$ plus non-vanishing explicit factors as in proposition \ref{proposition:1}, and then show that $\tau_{_{LU}}$ is related to $\tau_{_{LDU}}$ up to explicit terms. We know that the Malgrange form for the \rh{} on $i\bb{R}$ \eqref{tau_iR} is :
\begin{gather}
d_{t} \log\tau_{_{i\bb{R}}} = \int_{i\bb{R}} \frac{dz}{2\pi i } \Tr\left[ \Theta_{-}^{-1} \Theta_{-}' \dot{J}J^{-1}  \right]. \label{Mal1}
\end{gather}
Similarly, the Malgrange form of the \rh{} on LDU \eqref{F_i}: $Y_{i+1} = Y_{i} F_{i}$ is
\begin{gather}
d_{t}\log\tau_{_{LDU}} = \sum_{i=1}^{3} \int_{l_{i}} \frac{dz}{2\pi i } \Tr\left[ Y_{i}^{-1} Y_{i}'  \dot{F}_{i} F_{i}^{-1} \right]. \label{MalLDU}
\end{gather}
\begin{proposition}\label{proposition:3}
The Malgrange forms for the \rh{}s on the contours $i\bb{R}$ (RHP:\ref{rhp:5}) and on $LDU$ (RHP:\ref{rhp:6}) are related as
\begin{gather}
d_{t}\log\tau_{_{i\bb{R}}} = d_{t} \log\tau_{_{LDU}} - \int_{i\bb{R}} \frac{dz}{2\pi i } \dot{\left(\frac{\cB}{\cA} \right)} \left( \cA \cC' - \cA' \cC \right).
\end{gather}
the functions $\cA$, $\cB$, $\cC$ are defined in \eqref{jump_entries}.
\end{proposition}
\begin{proof}
We begin by substituting \eqref{LDU}: $J = F_{1} F_{2} F_{3}$ in the term
\begin{gather}
\dot{J}J^{-1}  = \left( \dot{F}_{1} F_{2} F_{3} + F_{1} \dot{F}_{2} F_{3} + F_{1} F_{2} \dot{F}_{3}\right) \left(F_{3}^{-1} F_{2}^{-1} F_{1}^{-1}  \right) \nonumber \\
 = \left( \dot{F}_{1} F_{1}^{-1} + F_{1} \dot{F}_{2} F_{2}^{-1} F_{1}^{-1} + F_{1} F_{2} \dot{F}_{3} F_{3}^{-1} F_{2}^{-1} F_{1}^{-1}  \right). \label{F123}
\end{gather}
Substituting in \eqref{F123} in the integrand of \eqref{Mal1},
\begin{gather}
\Tr\left[ \Theta_{-}^{-1} \Theta_{-}' \dot{J}J^{-1}  \right] = \Tr\left[ \Theta_{-}^{-1} \Theta_{-}' \left( \dot{F}_{1} F_{1}^{-1} + F_{1} \dot{F}_{2} F_{2}^{-1} F_{1}^{-1} + F_{1} F_{2} \dot{F}_{3} F_{3}^{-1} F_{2}^{-1} F_{1}^{-1}  \right)  \right]. \label{415}
\end{gather}
The equivalence \eqref{ide} along with the jump condition \eqref{F_i} imply that
\begin{gather} 
\Theta_{-} = Y_{1} \, , \quad \Theta_{-} F_{1} = Y_{2} \, ,\quad \Theta_{-} F_{1} F_{2} = Y_{3}. \label{idi}
\end{gather}
Substituting \eqref{idi} in \eqref{415},
\begin{gather}
\Tr\left[ \Theta_{-}^{-1} \Theta_{-}' \left( \dot{F}_{1} F_{1}^{-1} + F_{1} \dot{F}_{2} F_{2}^{-1} F_{1}^{-1} + F_{1} F_{2} \dot{F}_{3} F_{3}^{-1} F_{2}^{-1} F_{1}^{-1}  \right)  \right] \nonumber \\
= \Tr\left[ \Theta_{-}^{-1} \Theta_{-}'  \dot{F}_{1} F_{1}^{-1} + F_{1}^{-1} \Theta_{-}^{-1} \Theta_{-}' F_{1} \dot{F}_{2} F_{2}^{-1}  + F_{2}^{-1} F_{1}^{-1} \Theta_{-}^{-1} \Theta_{-}' F_{1} F_{2} \dot{F}_{3} F_{3}^{-1}    \right] \nonumber \\
= \Tr\left[ Y_{1}^{-1} Y_{1}'  \dot{F}_{1} F_{1}^{-1} + Y_{2}^{-1} Y_{2}' \dot{F}_{2} F_{2}^{-1} - F_{1}^{-1} F_{1}' \dot{F}_{2} F_{2}^{-1}+ Y_{3}^{-1} Y_{3}' \dot{F}_{3} F_{3}^{-1} - (F_{1} F_{2})^{-1} (F_{1} F_{2})' \dot{F}_{3} F_{3}^{-1}   \right] \nonumber \\
= \sum_{i=1}^{3} \Tr\left[ Y_{i}^{-1} Y_{i}'  \dot{F}_{i} F_{i}^{-1} \right] - \Tr\left[F_{1}^{-1} F_{1}' \dot{F}_{2} F_{2}^{-1} + (F_{1} F_{2})^{-1} (F_{1} F_{2})' \dot{F}_{3} F_{3}^{-1}   \right].
\end{gather}
Therefore, 
\begin{gather}
\int_{i\bb{R}} \frac{dz}{2\pi i} \Tr\left[ \Theta_{-}^{-1} \Theta_{-}' \dot{J}J^{-1}  \right] = \sum_{i=1}^{3} \int_{l_{i}} \frac{dz}{2\pi i} \Tr\left[ Y_{i}^{-1} Y_{i}'  \dot{F}_{i} F_{i}^{-1} \right] - \int_{i\bb{R}} \frac{dz}{2\pi i} \Tr\left[F_{1}^{-1} F_{1}' \dot{F}_{2} F_{2}^{-1} \right] \nonumber \\
\hspace*{0.5cm}- \int_{i\bb{R}} \frac{dz}{2\pi i} \Tr\left[ (F_{1} F_{2})^{-1} (F_{1} F_{2})' \dot{F}_{3} F_{3}^{-1}   \right]. \label{iR_LDU}
\end{gather}
Let us analyze the explicit terms.
\begin{itemize}
\item Since $F_{1}$ is upper triangular with constant diagonal entries, and $F_{2}$ is diagonal as defined in \eqref{LDU}, 
\begin{equation}
\Tr\left[F_{1}^{-1} F_{1}' \dot{F}_{2} F_{2}^{-1}\right]=0.
\end{equation} 
\item Substituting $F_{1,2,3}$ in the last term in \eqref{iR_LDU},
\begin{equation}
\Tr\left[(F_{1} F_{2})^{-1} (F_{1} F_{2})' \dot{F}_{3} F_{3}^{-1}   \right] = \dot{\left(\frac{\cB}{\cA} \right)} \left( \cA \cC' - \cA' \cC \right)
\end{equation}
where $\cA$, $\cB$, $\cC$ are explicit in terms of parabolic cylinder functions \eqref{jump_entries}.
\end{itemize}

Therefore, 
\begin{gather}
d_{t}\log\tau_{_{i\bb{R}}} = d_{t} \log\tau_{_{LDU}} - \int_{i\bb{R}} \frac{dz}{2\pi i } \dot{\left(\frac{\cB}{\cA} \right)} \left( \cA \cC' - \cA' \cC \right). \label{final2}
\end{gather}
\end{proof}

Recall from proposition \ref{proposition:3}, the Malgrange form of the \rh{} on LDU \eqref{MalLDU}:
\begin{gather}
d_{t}\log\tau_{_{LDU}} = \sum_{i=1}^{3} \int_{l_{i}} \frac{dz}{2\pi i } \Tr\left[ Y_{i}^{-1} Y_{i}'  \dot{F}_{i} F_{i}^{-1} \right]. \label{Mal3}
\end{gather}
For the \rh{} on LU (RHP:\ref{rhp:8}) with the jump condition \eqref{tF_1}: $\wY_{i+1} = \wY_{i} \wF_{i}$ where $i=1,3$, the Malgrange form reads
\begin{gather}
d_{t}\log\tau_{_{LU}} = \sum_{i=1,3} \int_{l_{i}} \frac{dz}{2\pi i } \Tr\left[ \wY_{i}^{-1} \wY_{i}'  \dot{\wF}_{i} \wF_{i}^{-1} \right]. \label{tau_LU}
\end{gather}
Refer to Appendix A \cite{bertola2012transition} for a proof of the integral form of the logarithmic derivative of the Fredholm determinant of the IIKS kernel.

\begin{proposition}\label{proposition:4} The Malgrange forms of the \rh{}s on contours LDU (RHP:\ref{rhp:6}) and LU (RHP:\ref{rhp:8}) are related as
\begin{gather}
d_{t} \log\tau_{_{LDU}} = d_{t} \log\tau_{_{LU}} + 2 \int_{i\bb{R}} \frac{dz}{2\pi i }\frac{\dot{\cA}(z,t)}{\cA(z,t)} \int_{i\bb{R}_{-}} \frac{dw}{2\pi i } \frac{\cA'(w,t)}{\cA(w,t)(z-w)}.
\end{gather}
\end{proposition}
\begin{proof}

We will first simplify the integrals on $l_{1}$ and $l_{3}$ in \eqref{Mal3}. Given that \eqref{wY}:$Y_{i} = \wY_{i} \varphi^{\sigma_{3}}$ and \eqref{tF}:$F_{i} = \varphi^{-\sigma_{3}}\wF_{i} \varphi^{\sigma_{3}}$,
\begin{gather}
\sum_{i=1,3} \int_{l_{i}} \frac{dz}{2 \pi i } \Tr\left[ Y_{i}^{-1} Y_{i}' \dot{F}_{i} F_{i}^{-1} \right] \nonumber \\
= \sum_{i=1,3} \int_{l_{i} } \frac{dz}{2 \pi i } \Tr\left[ \left(\wY_{i} \varphi^{\sigma_{3}}\right)^{-1} \left(\wY_{i} \varphi^{\sigma_{3}}\right)' \partial_{t}\left(\varphi^{-\sigma_{3}}\wF_{i}\varphi^{\sigma_{3}}\right) \left( \varphi^{-\sigma_{3}}\wF_{i}^{-1}\varphi^{\sigma_{3}}\right) \right] \nonumber \\
= \sum_{i=1,3} \int_{l_{i}}  \frac{dz}{2 \pi i } \Tr\left[ \left(\wY_{i}^{-1} \wY_{i}' + \sigma_{3} \varphi' \varphi^{-1}\right) \left(- \sigma_{3}\dot{\varphi} \varphi^{-1} + \dot{\wF}_{i} \wF_{i}^{-1}+ \wF_{i} \sigma_{3}\dot{\varphi} \varphi^{-1} \wF_{i}^{-1}\right) \right] \nonumber \\
= \sum_{i=1,3} \int_{l_{1}}\frac{dz}{2 \pi i } \Tr\left[ \wY_{i}^{-1} \wY_{i}' \dot{\wF}_{i} \wF_{i}^{-1} \right] +  \sum_{i=1,3}\int_{l_{i}} \frac{dz}{2 \pi i } \Tr\left[ \left(\wY_{i}^{-1} \wY_{i}' \right) \left(- \sigma_{3} \dot{\varphi} \varphi^{-1} + \wF_{i} \sigma_{3} \dot{\varphi} \varphi^{-1} \wF_{i}^{-1}\right) \right] \nonumber \\ 
+ \sum_{i=1,3}\int_{l_{i}} \frac{dz}{2\pi i} \Tr\left[\sigma_{3}  \varphi' \varphi^{-1} \left( - \sigma_{3}\dot{\varphi} \varphi^{-1} + \dot{\wF}_{i} \wF_{i}^{-1}+ \wF_{i}\sigma_{3} \dot{\varphi} \varphi^{-1} \wF_{i}^{-1} \right)  \right]. \label{LDU_LU1}
\end{gather}
In \eqref{LDU_LU1}, $\wF_{i}$ are either lower or upper triangular with constant diagonal entries as in \eqref{tF}. Therefore,
\begin{equation}
\int_{l_{1} \cup l_{3}} \frac{dz}{2\pi i} \Tr\left[ \sigma_{3} \varphi' \varphi^{-1} \left( - \sigma_{3} \dot{\varphi} \varphi^{-1} + \dot{\wF}_{i} \wF_{i}^{-1}+ \wF_{i}\sigma_{3} \dot{\varphi} \varphi^{-1} \wF_{i}^{-1} \right)  \right] =0.
\end{equation}
Therefore, given \eqref{tau_LU}, \eqref{LDU_LU1} reads
\begin{gather}
\sum_{i=1,3} \int_{l_{i} } \frac{dz}{2 \pi i } \Tr\left[ Y_{i}^{-1} Y_{i}' \dot{F}_{i} F_{i}^{-1} \right]= \sum_{i=1,3}\int_{l_{i}}\frac{dz}{2 \pi i } \Tr\left[ \wY_{i}^{-1} \wY_{i}' \dot{\wF}_{i} \wF_{i}^{-1} \right] \nonumber\\ 
+  \sum_{i=1,3}\int_{l_{i} } \frac{dz}{2 \pi i } \Tr\left[ \left(\wY_{i}^{-1} \wY_{i}' \right) \left(- \sigma_{3}\dot{\varphi} \varphi^{-1} + \wF_{i} \sigma_{3} \dot{\varphi} \varphi^{-1} \wF_{i}^{-1}\right) \right] \nonumber \\
=d_{t}\log\tau_{_{LU}}  + \sum_{i=1,3} \int_{l_{i}} \frac{dz}{2 \pi i } \Tr\left[ \left(\wY_{i}^{-1} \wY_{i}' \right) \left(- \sigma_{3} \dot{\varphi} \varphi^{-1} + \wF_{i} \sigma_{3} \dot{\varphi} \varphi^{-1} \wF_{i}^{-1}\right) \right]. \label{LU1}
\end{gather}
Recalling \eqref{tF_1}: $\wY_{i+1} = \wY_{i} \wF_{i}$, the second term in \eqref{LU1} can be further simplified 
\begin{gather}
 \sum_{i=1,3}\int_{l_{i} } \frac{dz}{2 \pi i } \Tr\left[ \left(\wY_{i}^{-1} \wY_{i}' \right) \left(- \sigma_{3} \dot{\varphi} \varphi^{-1} + \wF_{i} \sigma_{3} \dot{\varphi} \varphi^{-1} \wF_{i}^{-1}\right) \right]\nonumber \\
  = \sum_{i=1,3} \int_{l_{i} } \frac{dz}{2 \pi i } \Tr\left[- \wY_{i}^{-1} \wY_{i}' \sigma_{3} \dot{\varphi} \varphi^{-1} + \wF_{i}^{-1}\wY_{i}^{-1} \wY_{i}' \wF_{i}\sigma_{3}  \dot{\varphi} \varphi^{-1}  \right] \nonumber \\
  = \sum_{i=1,3}\int_{l_{i} } \frac{dz}{2 \pi i } \Tr\left[- \wY_{i}^{-1} \wY_{i}'  \sigma_{3} \dot{\varphi} \varphi^{-1} + \wY_{i+1}^{-1} \wY_{i+1}' \sigma_{3}  \dot{\varphi} \varphi^{-1} - \wF_{i}^{-1} \wF_{i}' \sigma_{3} \dot{\varphi} \varphi^{-1}  \right] \nonumber \\
  = \sum_{i=1,3}\int_{l_{i} } \frac{dz}{2 \pi i } \Tr\left[\Delta\left( \wY_{i}^{-1} \wY_{i}' \right) \sigma_{3} \dot{\varphi} \varphi^{-1} - \wF_{i}^{-1} \wF_{i}' \sigma_{3} \dot{\varphi} \varphi^{-1}  \right] \nonumber \\
  =\sum_{i=1,3} \int_{l_{i} } \frac{dz}{2 \pi i } \Tr\left[\Delta\left( \wY_{i}^{-1} \wY_{i}' \right) \sigma_{3} \dot{\varphi} \varphi^{-1} \right], \label{434}
\end{gather}
where now $\Delta\left( \wY_{i}^{-1} \wY_{i}' \right) =  \wY_{i+1}^{-1} \wY_{i+1}' - \wY_{i}^{-1} \wY_{i}' $. The last line is obtained using the fact that $ \Tr\left[\wF_{i}^{-1} \wF_{i}' \s \dot{\varphi} \varphi^{-1} \right]=0 $ since $\wF_{i}$ is either lower or upper triangular with constant diagonal entries, and $\varphi$ is scalar. 

The final expression in \eqref{434} can be further simplified by noting that the function $\varphi$ has no jumps on $l_{1}$ and $l_{3}$. Beginning with the integral on $l_{1}$,
\begin{gather}
\int_{l_{1}}  \frac{dz}{2 \pi i } \Tr\left[\Delta\left( \wY_{1}^{-1} \wY_{1}' \right) \sigma_{3} \dot{\varphi} \varphi^{-1} \right]  = \int_{l_{1}}  \frac{dz}{2 \pi i } \Tr\left[\left( \wY_{2}^{-1} \wY_{2}' - \wY_{1}^{-1} \wY_{1}' \right)\sigma_{3}  \dot{\varphi} \varphi^{-1} \right] \nonumber \\
= \int_{l_{1}}  \frac{dz}{2 \pi i } \Tr\left[ \wY_{2}^{-1} \wY_{2}'  \sigma_{3} \dot{\varphi} \varphi^{-1} \right]. \label{intl_1}
\end{gather}  
To obtain the last line, we notice from fig. \ref{fig:4} that $\int_{l_{1}}  \frac{dz}{2 \pi i } \Tr\left[ \wY_{1}^{-1} \wY_{1}' \sigma_{3} \dot{\varphi} \varphi^{-1} \right] =0$ by closing the contour on the right. A similar computation follows for the integral on $l_{3}$ in \eqref{434}
\begin{gather}
\int_{l_{3}}  \frac{dz}{2 \pi i } \Tr\left[\Delta\left( \wY_{3}^{-1} \wY_{3}' \right) \sigma_{3} \dot{\varphi} \varphi^{-1} \right]  = \int_{l_{3}}  \frac{dz}{2 \pi i } \Tr\left[\left( \wY_{4}^{-1} \wY_{4}' - \wY_{3}^{-1} \wY_{3}' \right) \sigma_{3} \dot{\varphi} \varphi^{-1} \right] \nonumber \\
= -\int_{l_{3}}  \frac{dz}{2 \pi i } \Tr\left[ \wY_{3}^{-1} \wY_{3}' \sigma_{3}  \dot{\varphi} \varphi^{-1} \right]. \label{intl_3}
\end{gather} 
To obtain the last line, we note that $\int_{l_{3}}  \frac{dz}{2 \pi i } \Tr\left[ \wY_{4}^{-1} \wY_{4}' \sigma_{3}  \dot{\varphi} \varphi^{-1} \right]=0$ by closing the contour on the left (see fig. \ref{fig:4}). 

Gathering the terms \eqref{intl_1}, \eqref{intl_3}, and using \eqref{tF_1}:$\wY_{2} = \wY_{3}$, \eqref{434} reads
\begin{gather}
\sum_{i = 1,3}\int_{l_{i} } \frac{dz}{2 \pi i } \Tr\left[\Delta\left( \wY_{i}^{-1} \wY_{i}' \right) \sigma_{3} \dot{\varphi} \varphi^{-1} \right] = \int_{l_{1}}  \frac{dz}{2 \pi i } \Tr\left[ \wY_{2}^{-1} \wY_{2}' \sigma_{3}  \dot{\varphi} \varphi^{-1} \right] -\int_{l_{3}}  \frac{dz}{2 \pi i } \Tr\left[ \wY_{3}^{-1} \wY_{3}'  \sigma_{3} \dot{\varphi} \varphi^{-1} \right] \nonumber \\
= -\int_{l_{2}}  \frac{dz}{2 \pi i } \Tr\left[ \wY_{2}^{-1} \wY_{2}'  \sigma_{3} \dot{\varphi} \varphi^{-1} \right] +\int_{l_{2}}  \frac{dz}{2 \pi i } \Tr\left[ \wY_{3}^{-1} \wY_{3}'  \sigma_{3} \dot{\varphi} \varphi^{-1} \right]
=0. \label{l13}
\end{gather}
Substituting \eqref{l13} in \eqref{LU1},
\begin{gather}
\sum_{i=1,3} \int_{l_{i} } \frac{dz}{2 \pi i } \Tr\left[ Y_{i}^{-1} Y_{i}' \dot{F}_{i} F_{i}^{-1} \right] = \sum_{i=1,3}\int_{l_{i}}\frac{dz}{2 \pi i } \Tr\left[ \wY_{i}^{-1} \wY_{i}' \dot{\wF}_{i} \wF_{i}^{-1} \right] = \partial_{t} \log\tau_{_{LU}}. \label{LDU_LU3}
\end{gather}
 
We now compute the integral on $l_{2}$ in \eqref{LDU_LU1}
\begin{gather}
\int_{l_{2}}\frac{dz}{2 \pi i } \Tr\left[ Y_{2}^{-1} Y_{2}' \dot{F}_{2} F_{2}^{-1} \right] = \int_{l_{2}}\frac{dz}{2 \pi i } \Tr\left[ \left(\wY_{2} \varphi_{-}^{\sigma_{3}}  \right)^{-1} \left(\wY_{2}' \varphi_{-}^{\sigma_{3}} + \wY_{2} \left(\varphi_{-}^{\sigma_{3}}\right)' \right) \dot{F}_{2} F_{2}^{-1} \right] \nonumber \\
= \int_{l_{2}}\frac{dz}{2 \pi i } \Tr\left[  \varphi_{-}^{-\sigma_{3}}  \wY_{2}^{-1} \left(\wY_{2}' \varphi_{-}^{\sigma_{3}} + \wY_{2} \left( \varphi_{-}^{\sigma_{3}}\right)' \right) \dot{F}_{2} F_{2}^{-1} \right] \nonumber \\
= \int_{l_{2}}\frac{dz}{2 \pi i } \Tr\left[    \wY_{2}^{-1} \wY_{2}'  \dot{F}_{2} F_{2}^{-1} + \sigma_{3} \varphi_{-}^{-1} \varphi_{-}' \dot{F}_{2} F_{2}^{-1}\right]. \label{intl_2} 
\end{gather}
Since $\wY_{2}$ does not jump on $l_{2}$, Liouville theorem implies that
\begin{gather}
\Tr\left[    \wY_{2}^{-1} \wY_{2}'  \dot{F}_{2} F_{2}^{-1} \right] =0. \label{Tzero}
\end{gather}
The term 
\begin{gather}
\Tr\left[\sigma_{3} \varphi_{-}^{-1} \varphi_{-}' \dot{F}_{2} F_{2}^{-1} \right] \label{T2}
\end{gather}
in \eqref{intl_2} is an explicit function of $\cA(z,w)$ in \eqref{jump_entries}. From \eqref{LDU},
\begin{gather}
F_{2} = \cA^{\sigma_{3}} \Rightarrow \dot{F}_{2} F_{2}^{-1} = \frac{\dot{\cA}}{\cA} \sigma_{3}. \label{f2}
\end{gather}
The function $\varphi_{-}$ is the boundary value of $\varphi$ defined in \eqref{varphi}
\begin{gather}
 \varphi_{-} =  \exp\left[\int_{i\bb{R}-\epsilon} \frac{dw}{2\pi i} \frac{\log \cA(w,t)}{z-w} \right] \Rightarrow \varphi_{-}^{-1} \varphi_{-}' = \int_{i\bb{R}-\epsilon} \frac{dw}{2\pi i } \frac{\cA'(w,t)}{\cA(w,t)(z-w)}. \label{vp_A}
\end{gather}
The expression \eqref{intl_2} simplifies as follows due to \eqref{f2}, \eqref{vp_A} 
\begin{gather} 
\int_{l_{2}}\frac{dz}{2 \pi i } \Tr\left[ Y_{2}^{-1} Y_{2}' \dot{F}_{2} F_{2}^{-1} \right] = \int_{l_{2}} \frac{dz}{2\pi i } \Tr\left[ \sigma_{3}\varphi_{-}^{-1} \varphi_{-}' \dot{F}_{2} F_{2}^{-1}  \right] \nonumber \\
=\int_{i\bb{R}} \frac{dz}{2\pi i } \left(2 \frac{\dot{\cA}(z,t)}{\cA(z,t)}\right) \int_{i\bb{R}- \epsilon} \frac{dw}{2\pi i } \frac{\cA'(w,t)}{\cA(w,t)(z-w)}. \label{TRACE}
\end{gather}
Substituting \eqref{TRACE} and \eqref{LDU_LU3}, in \eqref{Mal3}
\begin{gather} \label{LU}
d_{t} \log\tau_{_{LDU}} = \sum_{i=1}^{3} \int_{l_{i} } \frac{dz}{2 \pi i } \Tr\left[ Y_{i}^{-1} Y_{i}' \dot{F}_{i} F_{i}^{-1} \right] \nonumber \\ = \sum_{i=1,3} \int_{l_{i} } \frac{dz}{2 \pi i } \Tr\left[ Y_{i}^{-1} Y_{i}' \dot{F}_{i} F_{i}^{-1} \right]+ \int_{i\bb{R}} \frac{dz}{2\pi i } \Tr\left[\sigma_{3} \varphi_{-}^{-1} \varphi_{-}' \dot{F}_{2} F_{2}^{-1}  \right]  \nonumber\\
=   d_{t} \log\tau_{_{LU}} + 2 \int_{i\bb{R}} \frac{dz}{2\pi i }\frac{\dot{\cA}(z,t)}{\cA(z,t)} \int_{i\bb{R}_{-}} \frac{dw}{2\pi i } \frac{\cA'(w,t)}{\cA(w,t)(z-w)} \label{final3}
 \end{gather} 
\end{proof}
\pagebreak 

\section{Proof of theorem \ref{theorem:1}}\label{section:5}

\begin{proof}

The propositions \ref{proposition:1}, \ref{proposition:3}, \ref{proposition:4} imply that the \tfn{}s $\tau_{_{\Sigma}}$ and $\tau_{_{LU}}$ are related through explicit factors, and the proposition \ref{proposition:2} expresses $\tau_{_{LU}}$ as a Fredholm determinant. Therefore, the \tfn{} of \pii{} equation defined in \eqref{def:1}
{\small \begin{gather}
d_{t} \log\tau_{_{PII}}
 \equiv d_{t} \log\tau_{_{\Sigma}} 
\mathop{=}^{ \eqref{final1}} d_{t} \log\tau_{_{i\bb{R}}} - \int_{i\bb{R}} \frac{dz}{2\pi i} \widetilde{\mathcal{F}}(z,t;\nu,h)   -\left[ \frac{4i\nu}{3} + \frac{2 \nu^2}{t}  \right]-2g\left(\nu, h; t\right)\nonumber \\
\mathop{=}^{ \eqref{final2}} d_{t} \log\tau_{_{LDU}} - \int_{i\bb{R}} \frac{dz}{2\pi i } \left\lbrace \dot{\left(\frac{\cB}{\cA} \right)} \left( \cA \cC' - \cA' \cC \right)  + \widetilde{\mathcal{F}}(\z,\x,t;\nu,h)\right\rbrace-\left[ \frac{4i\nu}{3} + \frac{2 \nu^2}{t}\right] -2g\left(\nu, h; t\right)\nonumber \\
\mathop{=}^{ \eqref{final3}} d_{t} \log\tau_{_{LU}} + \int_{i\bb{R}} \frac{dz}{2\pi i } \left\lbrace \frac{2 \dot{\cA}(z,t)}{\cA(z,t)} \left(\int_{i\bb{R}-\epsilon} \frac{dw}{2\pi i } \frac{\cA'(w,t)}{\cA(w,t)(z-w)} \right)   - \dot{\left(\frac{\cB}{\cA} \right)} \left( \cA \cC' - \cA' \cC \right) 
 - \widetilde{\mathcal{F}}(z,t;\nu,h)\right\rbrace \nonumber \\ -\left[ \frac{4i\nu}{3} + \frac{2 \nu^2}{t} \right]-2g\left(\nu, h; t\right) \nonumber \\
 \mathop{=}^{ \eqref{result}} d_{t} \log\det\left[ \bb{1}_{_{L^2(i\bb{R})}}-\widetilde{\mathcal{K}} \right]-\left[ \frac{4i\nu}{3} + \frac{2 \nu^2}{t} \right]-2g\left(\nu, h; t\right) \nonumber\\
  +  \int_{i\bb{R}} \frac{dz}{2\pi i } \left\lbrace \frac{2 \dot{\cA}(z,t)}{\cA(z,t)} \left(\int_{i\bb{R}-\epsilon} \frac{dw}{2\pi i } \frac{\cA'(w,t)}{\cA(w,t)(z-w)} \right)  - \dot{\left(\frac{\cB}{\cA} \right)} \left( \cA \cC' - \cA' \cC \right) 
 - \widetilde{\mathcal{F}}(z,t;\nu,h)\right\rbrace . \label{FIN}
\end{gather}}
In \eqref{FIN}, the functions $\cA$, $\cB$, $\cC$ defined in \eqref{jump_entries} are explicit in terms of parabolic cylinder functions,  $\widetilde{\mathcal{F}}$ is defined in \eqref{F_def}, and the term 
 \begin{gather}
 \int_{i\bb{R}} \frac{dz}{2\pi i } \left\lbrace \frac{2 \dot{\cA}(z,t)}{\cA(z,t)} \left(\int_{i\bb{R}-\epsilon} \frac{dw}{2\pi i } \frac{\cA'(w,t)}{\cA(w,t)(z-w)} \right)  - \dot{\left(\frac{\cB}{\cA} \right)} \left( \cA \cC' - \cA' \cC \right) 
 - \widetilde{\mathcal{F}}(z,t;\nu,h)\right\rbrace
 \end{gather}
 depends only on $h$, $\nu$ and $t$. We then define 
 {\small \begin{align}
 \mathcal{F}(t,\nu,h) &:= \int_{i\bb{R}} \frac{dz}{2\pi i } \left\lbrace \frac{2 \dot{\cA}(z,t)}{\cA(z,t)} \left(\int_{i\bb{R}-\epsilon} \frac{dw}{2\pi i } \frac{\cA'(w,t)}{\cA(w,t)(z-w)} \right)  - \dot{\left(\frac{\cB}{\cA} \right)} \left( \cA \cC' - \cA' \cC \right) 
 - \widetilde{\mathcal{F}}(z,t;\nu,h)\right\rbrace\nonumber\\
 & \hspace{10.5cm} -2g\left(\nu, h; t\right) . \label{calF_def}
 \end{align}}
In terms of $\mathcal{F}(t,\nu,h)$, \eqref{FIN} reads
\begin{gather}
 d_{t} \log\tau_{_{PII}} = d_{t} \log\det\left[ \bb{1}_{_{L^2(i\bb{R})}} -\widetilde{\mathcal{K}} \right] + \mathcal{F}(t,\nu,h) -\left[ \frac{4i\nu}{3} + \frac{2 \nu^2}{t} \right].
 \end{gather}  
Therefore, the \tfn{} of \pii{} can be expressed as a Fredholm determinant of an integrable operator up to explicit factors. Furthermore, solving the RHP:\ref{rhp:8} is equivalent to solving the RHP:\ref{rhp:5}, which in turn is tantamount to solving the RHP:\ref{rhp:4}. Therefore, the zeros of $\tau_{_{PII}}$ (solvability condition of RHP:\ref{rhp:4}) are completely determined by the zeros of the Fredholm determinant \eqref{result}.
\end{proof}

\pagebreak
\nocite{*}
\printbibliography

\end{document}